\documentclass[manuscript,nonacm]{acmart}

\usepackage[utf8]{inputenc}
\usepackage[english]{babel}
\usepackage{mathrsfs}
\usepackage{multicol}
\usepackage{mathtools}
\usepackage{hyperref}
\usepackage{cleveref}
\usepackage{xspace}
\usepackage{mathrsfs}
\usepackage{amsmath}
\usepackage{amssymb}
\usepackage{amsthm}
\usepackage{esint}
\usepackage{tikz}
\usepackage{algpseudocode}
\usepackage{algorithm}
\usepackage{anyfontsize}
\usepackage{tcolorbox}
\usepackage{color}
\newtheorem{theorem}{Theorem}
\newtheorem{lemma}[theorem]{Lemma}
\theoremstyle{definition}
\newtheorem{definition}[theorem]{Definition}

\newcommand\Vol[0]{\operatorname{Vol}} 
\newcommand\polylog[0]{\ensuremath{\operatorname{polylog}}} 
\newcommand\clique{{\sc Congested Clique}\xspace}
\newcommand\congest{{\sc Congest}\xspace}

\newcommand\rem[0]{^{\textsf{r}}} 
\newcommand\suin[0]{_{\operatorname{in}}} 
\newcommand\suout[0]{_{\operatorname{out}}} 
\newcommand\suhi[0]{_{\operatorname{hi}}} 
\newcommand\vclusterof[1]{V_{#1}} 
\newcommand\vinnerof[1]{V_{#1}^\circ} 
\newcommand\vlistof[1]{V_{#1}^{-}} 
\newcommand\vgeqavgof[1]{V_{#1}^{*}} 
\newcommand\vhdof[1]{\vgeqavgof{#1}} 
\newcommand\vldof[1]{V_{#1}^{\operatorname{L}}} 
\newcommand\vlistldof[1]{V_{#1}^{<}} 
\newcommand\vcluster[0]{{\vclusterof{C}}}
\newcommand\vinner[0]{{\vinnerof{C}}}
\newcommand\vlist[0]{{\vlistof{C}}}
\newcommand\vgeqavg[0]{{\vgeqavgof{C}}}
\newcommand\vhd[0]{{\vhdof{C}}}
\newcommand\vld[0]{{\vldof{C}}}
\newcommand\vlistld[0]{\vlistldof{C}}
\newcommand\comdeg[1]{\ensuremath{\deg_C(#1)}} 
\newcommand\sentdeg[1]{\ensuremath{\deg^*_C(#1)}} 
\newcommand\ebar[0]{\ensuremath{\bar E}}
\newcommand\ep[0]{\ensuremath{E'}}
\newcommand\chainmem[3]{\ensuremath{\vchain{#1}#3[#2]}}
\newcommand\chainass[3]{\ensuremath{f_{\vchain{#1}#3}(#2)}}
\newcommand\chainres[3]{\ensuremath{f_{\vchain{#1}#3}^{-1}(#2)}}
\newcommand\Sim[0]{X} 
\newcommand\readop[0]{{\tt READ}\xspace} 
\newcommand\writeop[0]{{\tt WRITE}\xspace} 
\newcommand\auxop[0]{{\tt GET-AUX}\xspace} 
\newcommand\nin[0]{\ensuremath{N\suin}\xspace}
\newcommand\nout[0]{\ensuremath{N\suout}\xspace}
\newcommand\paramq[0]{\ensuremath{B_{\operatorname{aux}}}\xspace}
\newcommand\paramy[0]{\ensuremath{B_{\operatorname{write}}}\xspace}
\newcommand\naux[0]{\ensuremath{\ell}\xspace}
\newcommand\auxt[0]{\ensuremath{\alpha}\xspace}
\newcommand\maintok[0]{main token\xspace}
\newcommand\auxtok[0]{auxiliary token\xspace}
\newcommand\consdeg[0]{{\tt DEG}\xspace} 
\newcommand\consupdeg[0]{{\tt UP\_DEG}\xspace} 
\newcommand\conssz[0]{{\tt SIZE}\xspace} 
\newcommand\consdegbb[0]{{\tt DEG\_2to2}\xspace} 
\newcommand\consupdegbb[0]{{\tt UP\_DEG\_{}2to2}\xspace} 
\newcommand\consdegba[0]{{\tt DEG\_2to1}\xspace} 
\newcommand\consdegaa[0]{{\tt DEG\_1to1}\xspace} 
\newcommand\consupdegaa[0]{{\tt UP\_DEG\_1to1}\xspace} 
\newcommand\consupdegab[0]{{\tt UP\_DEG\_1to2}\xspace} 
\newcommand\specialalgo{partial-pass streaming algorithm\xspace}
\newcommand\Specialalgo{Partial-pass streaming algorithm\xspace}
\newcommand\SpecialAlgo{Partial-Pass Streaming Algorithm\xspace}
\newcommand\specialalgos{\specialalgo{}s\xspace}
\newcommand\Specialalgos{\Specialalgo{}s\xspace}
\newcommand\SpecialAlgos{\SpecialAlgo{}s\xspace}
\newcommand\conscont{input contiguity}
\newcommand\vchain[1]{\ensuremath{\mathcal{#1}}}
\newcommand\algo[1]{{\tt #1}}
\newcommand\rc[0]{\ensuremath{+ o(1)}}
\newcommand\rct[0]{\ensuremath{n^{o(1)}}}
\newcommand\rcf[0]{\ensuremath{\cdot \rct}}

\newcommand\genericCluster[0]{\ensuremath{C}}
\newcommand\genericClusterNodes[0]{\ensuremath{V_\genericCluster}}
\newcommand\genericClusterEdges[0]{\ensuremath{E_\genericCluster}}

\begin{abstract}
    The importance of classifying connections in large graphs has been
    the motivation for a rich line of work on distributed subgraph
    finding that has led to exciting recent breakthroughs. A crucial
    aspect that remained open was whether deterministic algorithms can
    be as efficient as their randomized counterparts, where the latter
    are known to be tight up to polylogarithmic factors.

    We give deterministic distributed algorithms for listing cliques
    of size $p$ in $n^{1 - 2/p + o(1)}$ rounds in the \congest model.
    For triangles, our $n^{1/3+o(1)}$ round complexity improves upon
    the previous state of the art of $n^{2/3+o(1)}$ rounds [Chang and
    Saranurak, FOCS 2020]. For cliques of size $p \geq 4$, ours are
    the first non-trivial deterministic distributed algorithms. Given known lower bounds, for all values $p \geq 3$ our
    algorithms are tight up to a $n^{o(1)}$ subpolynomial factor, which comes
    from the deterministic routing procedure we use.
\end{abstract}

\begin{document}
\title{Deterministic Near-Optimal Distributed Listing of Cliques}
\author{Keren Censor-Hillel}
 \email{ckeren@cs.technion.ac.il}
 \affiliation{\institution{Technion}\country{Israel}}
 \author{Dean Leitersdorf}
 \email{dean.leitersdorf@gmail.com}
 \affiliation{\institution{Technion}\country{Israel}}
 \author{David Vulakh}
 \email{dvulakh@mit.edu}
 \affiliation{\institution{Massachusetts Institute of Technology}\country{USA}}
\date{}
\maketitle

\section{Introduction}%
\label{sec:intro}

Subgraph listing is a crucial component of graph algorithms in many
computation settings, as it characterizes the connections in a graph.
In distributed computing, further motivation for finding small
subgraphs is that some algorithms run faster on networks without certain
subgraphs: \cite{TriangleFreeLocal1,TriangleFreeLocal2} show fast
algorithms for coloring and finding large cuts, respectively, in
triangle-free graphs, and \cite{CFG+21} quickly computes exact girth
in graphs without small cycles.

In this paper, we focus on clique listing in the \congest model of
distributed computing, in which $n$ vertices in a synchronous network
are required to find all cliques in the network graph by exchanging
messages of $O(\log n)$ bits.

A major open question about the \congest model is whether
randomization is necessary for this problem. A recent groundbreaking
line of work~\cite{ILG17,CPZ19,CS19,EFF+19,CHGL20,CCGL20} culminated
in clique-listing algorithms whose complexities are within
polylogarithmic factors of the lower bounds presented
by~\cite{PRS16,ILG17,FGKO18}. While these are all randomized
algorithms, powerful recent work by Chang and Saranurak~\cite{CS20}
showed the first non-trivial deterministic algorithm for the case of
triangle listing. Still, it exceeds the lower bound by a polynomial
number of rounds, leaving the following major question open (see also
Open Problem 2.2 in a recent survey~\cite{C21}):
\begin{tcolorbox}
    What is the complexity of deterministic clique listing in the
    \congest model, and is there a separation between randomized and
    deterministic algorithms?
\end{tcolorbox}

We present an algorithm answering this question for all $p$-cliques
(denoted $K_p$). Our algorithm finishes in $n^{1 - 2/p + o(1)}$
rounds, within $n^{o(1)}$ rounds of the $\tilde{\Omega}(n^{1 - 2/p})$
lower bound of~\cite{PRS16,ILG17,FGKO18}, and for listing cliques
larger than triangles is the first non-trivial deterministic
algorithm. This complexity matches the optimal randomized results of
\cite{CPSZ,CCGL20}, up to the additional $n^{o(1)}$ factor. The
additional $n^{o(1)}$ factor is due to the deterministic expander
routing scheme of~\cite{CS20}, and might be improved upon given future
progress in routing.

\begin{theorem} 
\label{thm:result}
    Given a constant $p \geq 3$ and a graph $G = (V, E)$ with
    $n = |V|$ vertices, there exists a deterministic \congest
    algorithm that completes in $n^{1 - 2/p \rc}$ rounds and lists all
    instances of $K_p$ in $G$.
\end{theorem}

Our key technical contribution which enables showing
Theorem~\ref{thm:result} is a simulation in \congest of a class
of algorithms we call \emph{\specialalgos}. We
use these to efficiently simulate deterministic load balancing tools
from the stronger \clique model (where every two vertices can
exchange $O(\log n)$-bit messages per round). It is possible that our
approach of extracting a \specialalgo and
simulating it could be useful elsewhere.

\subsection{Background and Challenges}

The groundbreaking triangle listing result of Chang, Pettie,
Saranurak, and Zhang~\cite{CPSZ}\footnote{Notice that \cite{CPSZ} is a
journal publication combining the two conference papers \cite{CPZ19,
CS19}.} consists of a distributed expander decomposition, in
which the graph is broken up into well-connected (high conductance /
low mixing time) clusters, and \clique-inspired algorithms that run
efficiently in each cluster using the expander routing
of~\cite{GhaffariKS17, GhaffariL18}. The remaining $\epsilon$-fraction
of inter-cluster edges are handled recursively. This general framework
has become the foundation for further work in this
area~\cite{EFF+19,CHGL20,CCGL20,CS20,CFG+21,IGM,CFLLO}.

Specifically, after running the expander decomposition, the $K_3$
listing algorithm of~\cite{CPSZ} runs a randomized version of the
triangle listing \clique algorithm of Dolev, Lenzen, and
Peled~\cite{DLP12}. For the case of $p \geq 4$, significant new
challenges arise as cliques of size $4$ and larger can have their
edges spread out over two or more clusters in the decomposition. This
implies that some edges need to be learned by a cluster although
neither of their endpoints is in the cluster, and thus the number of
edges that a cluster has to process is much larger. The listing
algorithms of Eden, Fiat, Fischer, Kuhn, and Oshman~\cite{EFF+19},
Censor-Hillel, Le Gall, and Leitersdorf~\cite{CHGL20}, and
Censor-Hillel, Chang, Le Gall, and Leitersdorf~\cite{CCGL20} use
various techniques for resolving the above.

\paragraph{\textbf{Usage of randomization in previous approaches.}}
Together, \cite{CPSZ,CCGL20} show optimal randomized algorithms for
listing cliques of any size in \congest. On a conceptual level,
randomization is used in three components of the algorithms. The first
two are the expander decomposition and the expander routing. The third
component that uses randomization is the load balancing for the
listing operation inside the clusters, upon which we now elaborate.

The load balancing task is as follows. Each cluster knows some
set of edges (distributed across the cluster vertices)~--- only the cluster's own edges for $K_3$, and for $K_{p\geq 4}$
perhaps also edges from other clusters. The goal is for the vertices to redistribute
these edges, such that any instance of $K_p$ composed of these edges
is known to at least one vertex in the cluster. This must happen without any vertex slowing down the network by sending or receiving too
much information. The standard approach is to construct a random
partition of the vertex set $V = V_1 \cup V_2 \cup \cdots \cup V_x$,
for $x = O(n^{1/p})$, so that the number of edges between two parts
$V_i,V_j$ is, with high probability, roughly
$|E|/x^2$. Then, each vertex $v$ is assigned some parts $V_i$ and
$v$ learns all edges between these parts. By assigning
a number of parts proportional to the degree of $v$, the near-optimal
round complexity is obtained.

\paragraph{\textbf{Challenges in removing randomization.}} Chang and Saranurak~\cite{CS20} show deterministic
expander decomposition and routing procedures. 
However, deterministically solving the above load balancing task is
difficult to do inside the clusters, implying a polynomial overhead to their triangle listing
round complexity.

The challenge of finding an efficient deterministic algorithm for the load balancing task arises not only in \congest, but also in \clique.  Specifically, it is unknown how to efficiently 
find in a deterministic manner a well-balanced partition of the graph such that the number of
edges between every two parts is roughly the same. In \clique, \cite{CTL20} bypassed this challenge for triangles, and \cite{CFG+21} were able to do so for any
subgraph (with at most $\log n$ vertices). The approach of the latter is that instead of using one
partition, they design a polynomial number of partitions and different
vertices choose to work with different combinations of
partitions. The different partitions are related to each other using a structure called partition trees.

A reasonable approach for deterministic load balancing in \congest for clique listing is to construct the above partition trees inside the clusters. However, the above algorithm for constructing these trees cannot directly be adapted to
\congest in an efficient manner. The reason for this is that the algorithm has
low round complexity, yet high message complexity: although it finishes in
$O(1)$ rounds in the \clique model, it involves sending $\Theta(n^2)$
messages. This high message complexity is very prohibitive for the \congest model, as the bandwidth available, i.e. the
number of communication edges in the graph or in a cluster, might be dramatically
smaller than $\Theta(n^2)$, as we now elaborate.

When dealing with $K_3$, our target round complexity is
$\tilde{O}(n^{1/3})$, as the known lower bound stated above is
$\tilde{\Omega}(n^{1/3})$. This implies a very hard restriction on the
bandwidth available within the clusters. Concretely, it can be shown that some low-degree vertices in the
clusters can only receive $O(n^{2/3})$ messages in total throughout
the entire algorithm. Therefore, for clusters with $\Omega(n^{2/3})$
vertices,
some vertices might not even know basic information about
the other cluster vertices, such as how much data each holds, and so performing any sort of load balancing is clearly tricky.

When dealing with $K_{p \geq 4}$, our target round complexity
inherently increases, as the lower bound is $\Omega(n^{1-2/p})$, which
is at least $\Omega(n^{1/2})$ for $p \geq 4$. While this eases the
pressure on the bandwidth, we are met with a new
challenge: constructing partition trees on data from multiple
clusters. As described above, a clique of size $\geq 4$ can have edges
which lie in multiple clusters, and so as in \cite{CCGL20}, we need to 
make some edges from outside a cluster known to the
vertices of the cluster. To capture this, we must substantially generalize the
definition of partition trees to distinguish between edges in the cluster and those whose information is brought into it.


\subsection{Our Approach}

We introduce a new
class of algorithms which we call \emph{\specialalgos}, and show that these require
very few messages to simulate in clusters in \congest. Then, we show
that there exists a \specialalgo which constructs a \emph{generalized
partition tree}. Finally, we plug this construction into the
deterministic expander decomposition and routing of \cite{CS20} in
order to get listing algorithms for all $K_p$.

\paragraph{\textbf{\Specialalgos.}} 
\Specialalgos are similar to standard streaming algorithms (see, for
example,
\cite{StreamingAlgorithmsFirstPapers1,StreamingAlgorithmsGodelPrizePaper}
or the survey \cite{StreamingAlgorithmsPopularSurvey}), as they maintain a small state, sequentially process an input
stream, and cannot return to previous parts of the input. However, a
key distinctive part of \specialalgos is that they strive to make less
than one pass over the input stream by skipping some of it.
In this setting, a stream is composed by tokens
reflecting two levels of granularity: a large number of
\emph{\auxtok{s}}, each with some fine-grained data, and a
small number of \emph{\maintok{s}} with very coarse-grained
data summarizing the \auxtok{s}. Each \maintok summarizes the chunk of 
\auxtok{s} that appears after it and before the next \maintok in the stream. 
\Specialalgos are able to access all the \maintok{s} in a given stream,
while reading only a limited amount of \auxtok{s}. As such, the total number of stream tokens accessed by a
\specialalgo is significantly smaller than the total stream length.

As intuition for this definition and why it fits well with a \congest simulation, recall that in \congest the input is dispersed across vertices. The raw data each vertex gets can be seen as \auxtok{s}, and each vertex can locally compute a summary of its data, creating a \maintok. When an algorithm runs, it might suffice to sometimes only read the \maintok produced by a vertex, while at other occasions the internal \auxtok{s} are needed.

\paragraph{\textbf{Simulating \specialalgos in \congest.}} For simplicity, suppose each
vertex has a single \maintok which summarizes possibly many \auxtok{s}
which it holds. Assume also that the vertices are ordered in the way
their tokens appear in the stream, i.e., vertex with identifier $1$
has the first \maintok. The algorithm which we run on the
stream has a small $\polylog(n)$-bit state. Consider
two extreme approaches.

\emph{Approach 1: State Passing.} Vertex $1$ starts simulating the algorithm and
locally executes it on its tokens. Then, it passes the state of the
algorithm to vertex $2$ and so forth. While only $\tilde\Theta(n)$
messages are used, this takes $\tilde\Theta(n)$ rounds.

\emph{Approach 2: Leader with Queries.} An arbitrary vertex $v$ is
denoted leader and learns all $n$ \maintok{s}.
Then, $v$ locally simulates the algorithm. When
\auxtok{s} of a certain \maintok are needed, $v$ sends the state of
the algorithm to the vertex $u$ which sent $v$ that specific \maintok. 
Vertex $u$ simulates the algorithm on its local \auxtok{s} and sends the
state back to $v$. As \maintok{s}
summarize the \auxtok{s} and \specialalgos are able to skip the
vast majority of \auxtok{s}, $v$ sends the state to other
vertices only few times during the entire execution. However, we are
burdened by the fact that $v$ initially learns all $n$
\maintok{s} by itself.

\emph{Our approach.} We opt for a combination of both approaches. Denote some set of
vertices $A$ of sublinear polynomial size, ordered $A = \{a_1, \dots, a_{|A|}\}$. 
The vertices in
$A$ collectively learn all \maintok{s}~--- $a_1$ learns the first $n/|A|$ \maintok{s} tokens, followed by $a_2$ which learns the next $n/|A|$, and so forth.
Vertex
$a_1$ starts executing the algorithm on its local data.
Whenever it hits a \maintok for which the corresponding \auxtok{s} are needed, it sends the state of the
algorithm back to whichever vertex sent it the relevant \maintok. That
vertex simulates the algorithm, and returns the state back to $a_1$.
After $a_1$ processes the \maintok{s} it knows, it forwards
the state of the algorithm to $a_2$, and so forth, until
$a_{|A|}$ completes the run of the algorithm. This approach both passes the state few times
and each vertex learns few \maintok{s}, thus having both low round and low message complexities.

\paragraph{\textbf{Constructing generalized partition trees by a \specialalgo.}} We design a \specialalgo which constructs the generalized
partition trees we require, so that we can then invoke our simulation of it in \congest. For clarity, the following is
a simplification of the full algorithm which we show. At a very high level, a
common operation in constructing partition trees is that there is a globally known
vertex set $W$, and we strive to partition the graph vertices $V$ into sets
$V_1, \dots, V_x$, for some $x$, such that for each $i, j \in [x]$,
the number of edges in $E(V_i, W)$ is roughly the same as that in
$E(V_j, W)$. 
We work within clusters and each cluster needs to process also information that is brought from outside it. Concretely, denoting $V_C$ as the set of vertices in a cluster $C$, each $v \in V_C$ might also hold some edges incident to vertices outside $C$, a set which we for now denotes as $F(v) \subseteq V \setminus V_C$.

The \specialalgo works as follows. For each $v \in V_C$, define
$|F(v)| + 1$ \auxtok{s} as the value $|E(\{v\}, W)|$ and values
$|E(\{u\}, W)|$ for each $u \in F(v)$. Also define a \maintok
for $v$ as the sum of these \auxtok{s}. The input stream is the concatenation of the tokens defined for each $v \in V_C$. The algorithm maintains an \emph{incomplete part} $V_j$ in the partition to which we
currently add vertices. When reaching the \maintok of vertex $v \in V_C$, we attempt
to add all of $F(v) \cup \{v\}$ to $V_j$. By observing the \maintok of
$v$, we know how many edges will be added to $E(V_j, W)$ if we do
so. If this value is acceptable, the algorithm proceeds. Otherwise, we
request all the \auxtok{s} of $v$, and add the vertices one by one
until $V_j$ grows too large, at which point we \emph{close} $V_j$ and
start filling part $V_{j+1}$.

\paragraph{\textbf{Roadmap.}} We provide required background in
Section~\ref{sec:preliminaries}, and in Section~\ref{sec:special-algo}
introduce and efficiently simulate \specialalgos in high-conductance
clusters in \congest. In Section~\ref{sec:partitionTrees}, we
construct partition trees using \specialalgos.
Sections~\ref{sec:triangles} and~\ref{sec:cliques} wrap these into
$K_3$ and $K_{\geq 4}$ listing algorithms, respectively.

\subsection{Further Related Work}

\textbf{Triangle finding.} In \clique, the seminal work of
\cite{DLP12} presents an $O(n^{1/3}/\log n)$-round, deterministic
algorithm for $K_3$ listing, and shows faster algorithms for detection
and counting for certain graph families. Matching lower bounds for
listing are shown in \cite{PRS16, ILG17}. For sparse graphs,
\cite{DLP12,PRS16,CTL20} show faster algorithms, some of which are
tight due to \cite{PRS16}. The state-of-the-art for triangle detection
and counting in \clique is $O(n^{0.158})$ rounds, based on fast matrix
multiplication \cite{CHKKLPS}. It is considered hard to show lower
bounds for this, due to the reduction to circuit complexity of
\cite{DKO}. the \emph{broadcast} version of \clique.

In \congest, the first breakthrough for triangle finding is due to
\cite{ILG17}, who showed the first non-trivial algorithms for listing
and detection, in $\tilde O(n^{3/4})$ and $\tilde O(n^{2/3})$ rounds
w.h.p, respectively. Subsequently, \cite{CPSZ} showed a breakthrough
by listing in the optimal complexity of $\tilde O(n^{1/3})$ rounds,
w.h.p., using an expander decomposition together with the routing
techniques of \cite{GhaffariKS17, GhaffariL18}. \cite{HPZZ} show an
algorithm which takes $O(\Delta/\log n + \log \log \Delta)$ rounds,
where $\Delta$ is the maximal degree in the graph. following:
Subsequently, \cite{CPZ19} showed a breakthrough by listing in
$\tilde O(n^{1/2})$ rounds, w.h.p., using an expander decomposition,
which was later improved to get the optimal, yet randomized,
$\tilde O(n^{1/3})$-round algorithm of \cite{CS19}. The first
non-trivial deterministic triangle finding algorithm was given by
\cite{CS20}, taking $O(n^{0.58})$ rounds for detection and
$n^{2/3+o(1)}$ rounds for listing. On the lower bound side, it is only
known that more than one round is needed for detection in \congest
\cite{ACKL,FGKO18}. Proving any polynomial lower bound is considered
difficult~\cite{EFF+19}.

\textbf{Finding larger cliques.} In \clique, \cite{DLP12}
provide an $O(n^{1-2/p}/\log n)$-round deterministic algorithm for
$K_p$ listing. Further, their algorithm lists all subgraphs with $p$
vertices within this complexity. This is optimal for cliques due to
the lower bound of \cite{FGKO18}. A randomized sparsity aware
algorithm was shown for $K_p$ listing, for $p \geq 4$, in
\cite{CHGL20}, and a deterministic construction using partition trees
was later given by \cite{CFG+21}; these results are tight, for graphs
of all sparsities, due to \cite{PRS16, ILG17,FGKO18}.

In \congest, the expander decompositions of \cite{CPZ19,CS19} were
first used by \cite{EFF+19} to show $K_4$ and $K_5$ listing in
$n^{5/6 + o(1)}$ and $n^{73/75+o(1)}$ rounds w.h.p., respectively.
Later, \cite{CHGL20} showed the first sub-linear result for all
$p \geq 4$, which was followed by the $\tilde O(n^{1-2/p})$-round
algorithm of \cite{CCGL20} --- this is optimal due to the lower bound
of \cite{FGKO18}. For $p = 4$, it turns out that the optimal
$\tilde \Theta(n^{1/2})$-round result for listing $K_4$ is also the
best possible for detecting $K_4$, due to the lower bound by
\cite{CK18}.

\textbf{Finding other subgraphs and additional related
results.} Small cycles received much attention with many algorithms
and hardness results, shown in \clique by \cite{CHKKLPS,CTL20,CFG+21},
and in \congest by \cite{DKO,KR17,CFG+21,FGKO18,EFF+19}. For all
$p \geq 3$, a gap between $K_p$ detection and listing is shown in
quantum variants of \congest and \clique in \cite{IGM,CFLLO}. Finding
trees was investigated in \cite{DKO,KR17}, and more general subgraphs
in \cite{GO17,FGKO18,EFF+19}. Distributed subgraph testing, in which
we need to determine if a graph is either free of a subgraph or many
edges need to be removed from it to reach this state, was studied in
\cite{BPS11,CHFSV18,ThreeNotes,FO17,FRST16}.


\section{Preliminaries}%
\label{sec:preliminaries}

\textbf{Notations.} Given a graph $G = (V, E)$, and a set
$E' \subseteq E$, denote by $G[E']$ the graph induced by $E'$ on $G$.
For a vertex $v \in V$ and a subgraph $G' = (V', E')$ of $G$, denote
by $\deg_{G'}(v)$ the degree of $v$ in $G'$,
$\deg_{G'} = |\{ u \in V' : (v, u) \in E' \}|$. For a subset of
vertices $V' \subseteq V$ and a vertex $v \in V$, denote by
$\deg_{V'}$ the degree of $v$ into $V$,
$\deg_{V'} = |\{ u \in V' : (v, u) \in E \}|$.

Decompositions into clusters with high conductance (low mixing time)
underpin \congest listing algorithms. Here, we restate the required
background about them.

\begin{definition}[Conductance]%
\label{def:conductance}
    Given a graph $G = (V, E)$ and a subset $S \subseteq V$, define
    $\Vol(S) = \sum_{v \in S} \deg(v)$ and
    $\partial(S) = E(S, V \setminus S)$. Then, the \emph{conductance}
    of a cut $S \subseteq V$ is
    $\Phi(S) = |\partial(S)| / \min\{\Vol(S), \Vol(V \setminus S)$ \}
    and the conductance of $G$ is the minimum conductance over all
    nontrivial cuts $S$,
    $
        \Phi(G) =
        \min_{S \subseteq V : S \ne \emptyset \land S \ne V} \Phi(S)
    $.
\end{definition}

A related value is \emph{mixing time}, denoted $\tau(G)$, which is
roughly the number of steps a random walk takes to reach the
stationary distribution~--- a precise definition is not required for
our purposes. As demonstrated by~\cite[Corollary 2.3]{JS89}, $\tau(G)$ and $\Phi(G)$
are related by $\tau(G) \le \Theta(\frac {\log n}{\Phi(G)^2})$. As $\tau(G)$
is always at least the graph diameter, this shows the following.

\begin{theorem}%
\label{thm:diameter-conductance}
The diameter of a graph with conductance $\phi$ is
$O(\phi^{-2}\log n)$.
\end{theorem}

\begin{definition}[$\phi$-Cluster, Expander Decomposition]%
\label{expander-cluster}%
\label{def:expander}
    A subgraph
    $\genericCluster = (\genericClusterNodes, \genericClusterEdges)$
    of a given graph $G = (V, E)$ is called a \emph{$\phi$-Cluster in
    $G$} if $\Phi(\genericCluster) \geq \phi$. An
    \emph{$(\epsilon, \phi)$-expander decomposition} of a graph $G$ is
    a partition of 
    $E = E_1 \cup E_2 \cup \cdots \cup E_x \cup E\rem$ for some $x$, 
    such that the subgraphs
        $
            G\left[ E_1 \right],
            \dots,
            G\left[ E_x \right]
        $
    are vertex-disjoint $\phi$-clusters in $G$, and
    $\left|E\rem\right| \leq \epsilon|E|$.
\end{definition}

The following are the two main deterministic tools that we need
from~\cite{CS20}.

\begin{theorem}[\congest Expander Decomposition {\hspace{-1pt}\cite[Theorem~1.1]{CS20}}]%
\label{thm:cs-decomposition}
    Let $0 < \epsilon < 1$ be a parameter. An
    $(\epsilon,\phi)$-expander decomposition of a graph $G = (V, E)$
    with 
    $
        \phi =
            \operatorname{poly}( \epsilon ) \cdot
            2^{-O( \sqrt{\log n \log\log n} )}
    $
    can be computed by a deterministic algorithm in \congest in
    $
        \operatorname{poly}( \epsilon^{-1} ) \cdot
            2^{O( \sqrt{\log n \log\log n} )}
    $
    rounds.
\end{theorem}

\begin{theorem}[\congest Routing {\hspace{-1pt}\cite[Theorem~1.2]{CS20}}]%
\label{thm:cs-routing}
    Let $G = (V, E)$ be a graph with conductance $\phi$, where each
    vertex $v \in V$ is a source and destination of
    $O(L) \cdot \deg(v)$ messages. Then there is a deterministic
    algorithm in the \congest that routes all messages to their
    destination in
    $
        L \cdot \operatorname{poly}( \phi^{-1} ) \cdot
            2^{O( \log^{2/3}n\log^{1/3}\log n )}
    $
    rounds.
\end{theorem}

In order to achieve sufficient bandwidth within clusters, we need both
good conductance and good guarantees about the degrees of vertices in
the cluster.
The following definition refines $\phi$-clusters with the needed
constraints.
Figure~\ref{fig:cluster-sets} illustrates the different designations
for vertices in such clusters.

\begin{definition}[$(\phi, \delta)$-Communication Cluster]%
\label{def:com-cluster}
    A \emph{$(\phi, \delta)$-communication cluster} is a
    $\phi$-cluster
    $\genericCluster = (\genericClusterNodes, \genericClusterEdges)$
    in a graph $G = (V, E)$, together with a given subset
    $\vlist \subseteq \vcluster$ such that for each $v \in \vlist$ it
    holds that $\comdeg v \ge \delta$, and $v$ knows the subset
    $\vlist$ and the values $\delta$, $n = |V|$,
    $K = |\genericClusterNodes|$, $k = |\vlist|$. The
    \emph{communication degree} of a vertex $v \in \vlist$ is the
    number of edges in $\genericClusterEdges$ that are adjacent to
    $v$,
    $\comdeg v = |\{u \in \genericClusterNodes : (v, u) \in E_C\}|$.
    Denote by $\mu$ the average communication degree of vertices in
    $\vlist$, that is, $\mu = \frac1k\sum_{v \in \vlist}\comdeg v$,
    and denote by $\vhd$ the set of vertices in $\vlist$ with at least
    half-average communication degree, that is,
    $\vhd = \{v \in \vlist : \comdeg v \ge \frac{1}{2}\mu\}$.
\end{definition}

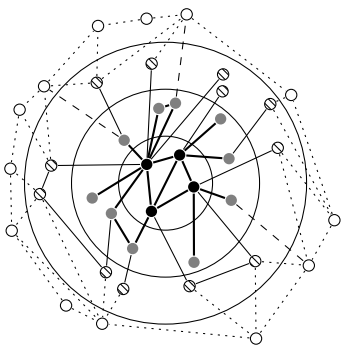
\begin{figure}%
\begin{multicols}{2}
\begin{center}
\begin{tikzpicture}[scale=0.25]
\usetikzlibrary{patterns}
\node[minimum size=1.5mm, inner sep=0, fill,circle] at (-1, 1) (1) {};
\node[minimum size=1.5mm, inner sep=0, fill,circle] at (-0.75, -1.5) (2) {};
\node[minimum size=1.5mm, inner sep=0, fill,circle] at (0.75, 1.5) (3) {};
\node[minimum size=1.5mm, inner sep=0, fill,circle] at (1.5, -0.2) (4) {};
\draw (0, 0) circle (2.5);
\node[minimum size=1.5mm, inner sep=0, fill,color=gray,circle] at (-2.88, -1.61) (5) {};
\node[minimum size=1.5mm, inner sep=0, fill,color=gray,circle] at (3.376, 1.304) (7) {};
\node[minimum size=1.5mm, inner sep=0, fill,color=gray,circle] at (-1.75, -3.48) (8) {};
\node[minimum size=1.5mm, inner sep=0, fill,color=gray,circle] at (2.930, 3.422) (9) {};
\node[minimum size=1.5mm, inner sep=0, fill,color=gray,circle] at (-0.36, 3.978) (10) {};
\node[minimum size=1.5mm, inner sep=0, fill,color=gray,circle] at (3.495, -0.91) (11) {};
\node[minimum size=1.5mm, inner sep=0, fill,color=gray,circle] at (0.530, 4.258) (12) {};
\node[minimum size=1.5mm, inner sep=0, fill,color=gray,circle] at (1.513, -4.21) (13) {};
\node[minimum size=1.5mm, inner sep=0, fill,color=gray,circle] at (-2.20, 2.284) (14) {};
\node[minimum size=1.5mm, inner sep=0, fill,color=gray,circle] at (-3.90, -0.79) (15) {};
\draw (0, 0) circle (5);
\node[minimum size=1.5mm, inner sep=0, draw,circle,pattern=north west lines] at (4.770, -4.15) (16) {};
\node[minimum size=1.5mm, inner sep=0, draw,circle,pattern=north west lines] at (5.967, 1.881) (17) {};
\node[minimum size=1.5mm, inner sep=0, draw,circle,pattern=north west lines] at (-6.08, 0.937) (18) {};
\node[minimum size=1.5mm, inner sep=0, draw,circle,pattern=north west lines] at (-2.25, -5.63) (21) {};
\node[minimum size=1.5mm, inner sep=0, draw,circle,pattern=north west lines] at (3.061, 5.789) (22) {};
\node[minimum size=1.5mm, inner sep=0, draw,circle,pattern=north west lines] at (5.570, 4.210) (23) {};
\node[minimum size=1.5mm, inner sep=0, draw,circle,pattern=north west lines] at (-0.74, 6.359) (24) {};
\node[minimum size=1.5mm, inner sep=0, draw,circle,pattern=north west lines] at (-6.68, -0.58) (25) {};
\node[minimum size=1.5mm, inner sep=0, draw,circle,pattern=north west lines] at (-3.17, -4.74) (26) {};
\node[minimum size=1.5mm, inner sep=0, draw,circle,pattern=north west lines] at (3.030, 4.888) (28) {};
\node[minimum size=1.5mm, inner sep=0, draw,circle,pattern=north west lines] at (1.277, -5.48) (29) {};
\node[minimum size=1.5mm, inner sep=0, draw,circle,pattern=north west lines] at (-3.66, 5.347) (30) {};
\draw (0, 0) circle (7.5);
\node[minimum size=1.5mm, inner sep=0, draw,circle] at (7.619, -4.38) (31) {};
\node[minimum size=1.5mm, inner sep=0, draw,circle] at (8.979, -1.97) (32) {};
\node[minimum size=1.5mm, inner sep=0, draw,circle] at (6.682, 4.695) (33) {};
\node[minimum size=1.5mm, inner sep=0, draw,circle] at (4.813, -8.27) (34) {};
\node[minimum size=1.5mm, inner sep=0, draw,circle] at (1.124, 8.985) (35) {};
\node[minimum size=1.5mm, inner sep=0, draw,circle] at (-6.46, 5.157) (36) {};
\node[minimum size=1.5mm, inner sep=0, draw,circle] at (-7.76, 3.910) (37) {};
\node[minimum size=1.5mm, inner sep=0, draw,circle] at (-3.37, -7.48) (39) {};
\node[minimum size=1.5mm, inner sep=0, draw,circle] at (-8.18, -2.53) (40) {};
\node[minimum size=1.5mm, inner sep=0, draw,circle] at (-3.59, 8.367) (41) {};
\node[minimum size=1.5mm, inner sep=0, draw,circle] at (-8.25, 0.772) (42) {};
\node[minimum size=1.5mm, inner sep=0, draw,circle] at (-5.29, -6.51) (43) {};
\node[minimum size=1.5mm, inner sep=0, draw,circle] at (-1.01, 8.753) (45) {};
\draw[thick] (1) edge (2);
\draw[thick] (1) edge (3);
\draw[thick] (2) edge (3);
\draw[thick] (2) edge (4);
\draw[thick] (3) edge (4);
\draw[thick] (1) edge (5);
\draw[thick] (3) edge (7);
\draw[thick] (2) edge (8);
\draw[thick] (3) edge (9);
\draw[thick] (1) edge (10);
\draw[thick] (4) edge (11);
\draw[thick] (1) edge (12);
\draw[thick] (4) edge (13);
\draw[thick] (1) edge (14);
\draw[thick] (1) edge (15);
\draw[thick] (10) edge (12);
\draw[thick] (8) edge (5);
\draw (4) edge (16);
\draw (4) edge (17);
\draw (1) edge (18);
\draw (8) edge (21);
\draw (1) edge (22);
\draw (7) edge (23);
\draw (1) edge (24);
\draw (18) edge (25);
\draw (5) edge (26);
\draw (3) edge (28);
\draw (2) edge (29);
\draw (14) edge (30);
\draw (26) edge (25);
\draw (29) edge (16);
\draw[dashed] (11) edge (31);
\draw[dashed] (12) edge (35);
\draw[dashed] (14) edge (36);
\draw[dotted] (16) edge (31);
\draw[dotted] (17) edge (31);
\draw[dotted] (32) edge (31);
\draw[dotted] (17) edge (32);
\draw[dotted] (23) edge (32);
\draw[dotted] (23) edge (33);
\draw[dotted] (32) edge (33);
\draw[dotted] (35) edge (33);
\draw[dotted] (31) edge (34);
\draw[dotted] (16) edge (34);
\draw[dotted] (29) edge (34);
\draw[dotted] (24) edge (35);
\draw[dotted] (30) edge (35);
\draw[dotted] (30) edge (36);
\draw[dotted] (18) edge (36);
\draw[dotted] (18) edge (37);
\draw[dotted] (36) edge (37);
\draw[dotted] (40) edge (39);
\draw[dotted] (25) edge (39);
\draw[dotted] (21) edge (39);
\draw[dotted] (26) edge (39);
\draw[dotted] (34) edge (39);
\draw[dotted] (43) edge (39);
\draw[dotted] (42) edge (40);
\draw[dotted] (43) edge (40);
\draw[dotted] (25) edge (40);
\draw[dotted] (45) edge (41);
\draw[dotted] (36) edge (41);
\draw[dotted] (30) edge (41);
\draw[dotted] (25) edge (42);
\draw[dotted] (37) edge (42);
\draw[dotted] (45) edge (35);
\end{tikzpicture}
\end{center}
\caption{%
    A schematic of the notations for the types of vertices and edges
    used in this paper. The rings partition the sets
    $\vhd \subseteq \vlist \subseteq \vcluster \subseteq V$, so black
    vertices are members of $\vhd$, grey ones of $\vlist$, hatched
    ones of $\vcluster$, and white ones of $V$. The bold edges are
    edges in $E(\vlist, \vlist)$; these are the edges for which we
    endeavor to list cliques. Solid edges are edges used for
    communication, $E(\vlist, \vcluster)$. Solid edges that cross the
    second ring and dashed edges are the elements of
    $\ebar \subseteq E(\vlist, V \setminus \vlist)$. Dotted edges are
    elements of
    $E' \subseteq E(V \setminus \vlist, V \setminus \vlist)$.
}%
\label{fig:cluster-sets}
\end{multicols}
\end{figure}

In Sections~\ref{sec:triangles} and~\ref{sec:cliques}, we use the
decomposition that is produced by Theorem~\ref{thm:cs-decomposition}
as the foundation of our recursive listing algorithm.
We produce a $(\phi,\delta)$-communication cluster $C_i$ from each
edge set $E_i$ returned by the decomposition,
then use the routing scheme of
Theorem~\ref{thm:cs-routing} to run algorithms for listing the cliques
that intersect each cluster.
To show that this plan works, it is necessary to show that $\phi$ and
$\delta$ are sufficiently high and that the number of edges we can
eliminate in each step of the recursion is at least a constant
fraction of the total.
The first
fact depends on the details of how we choose $C_i$, which differ
for different sizes of cliques are and are discussed in
Sections~\ref{sec:triangles} and~\ref{sec:cliques}.
Here, we show a strategy for selecting a set of edges for which we
attempt to list all cliques containing at least one edge in that set
and state a lemma of~\cite{CS20} that shows that
a constant fraction of
the total edges can be eliminated in this way.

Recall that Theorem~\ref{thm:cs-decomposition} partitions the set of
edges of the graph $E$ into
$E = E_1 \cup E_2 \cup \cdots E_x \cup E\rem$, where
$|E\rem| \le \epsilon|E|$.
Consider the subgraph induced on $G$ by
each edge set $E_i$, $G[E_i] = (V_i, E_i)$.
For each $E_i$, we select a set of vertices
$\vinnerof{C_i}$ and seek to use the cluster derived from $E_i$ to
list cliques containing an edge in
$E(\vinnerof{C_i}, \vinnerof{C_i})$.
Let $\vinnerof{C_i}$ be the set of
vertices in $V_i$ that have the majority of their edges in $E_i$, that
is,
$\vinnerof{C_i} = \{v \in V_i : \deg_{E_i}(v) \ge \deg_{E \setminus E_i}(v)\}$.
Also denote by $E_i^-$ the set of edges in $E_i$ that are between two
vertices in $\vinnerof{C_i}$, that is,
$E_i^- = \{ e = (u, v) \in E_i : \{ u, v \} \subseteq \vinnerof{C_i} \}$.
If we know that the
number of remaining edges not contained in some $E_i$ is sufficiently
small, we can recursively repeat the algorithm on the graph induced by
the set of remaining edges and obtain a logarithmic recursion depth.
The following lemma states the property we need; as a corollary, the
total number of edges in $E$ that are not in some $E_i^-$ is at most
$3\epsilon|E|$.

\begin{lemma}[{\hspace{-1pt}\cite[Lemma 6.1]{CS20}}]%
\label{lem:remaining-edges-small}
    It holds that
    $
        |\bigcup_{1 \leq i \leq x} E_i \setminus E_i^-|
        \leq 2\epsilon|E|
    $.
\end{lemma}

\section{Simulating \SpecialAlgos in \congest}%
\label{sec:special-algo}

A \emph{\specialalgo} for parameters $L,\nin,\nout,\paramq,\paramy$ is
a procedure that satisfies the following. There is an \emph{input
stream} $S$, which is a sequence of $\nin$ \emph{\maintok{s}},
$S = \langle \tau_1, \tau_2, \dots, \tau_{\nin} \rangle$. Each
\maintok $\tau_i$ has an associated sequence of $\naux_i$
\emph{\auxtok{s}},
$
    \auxt_i = \langle
        \auxt_{i,1}, \auxt_{i,2}, \dots, \auxt_{i,\naux_i}
    \rangle
$
The algorithm produces a \emph{write-only output stream}
$R = \langle \rho_1,\rho_2,\dots,\rho_{\nout} \rangle$, by invoking
operations that are either \readop, \auxop, or \writeop, as follows.
The operation \readop reads the next token from the input stream $S$;
the operation \auxop adds the \auxtok{s} associated with the last-read
\maintok to the beginning of $S$, so that if the previous \readop
operation returned $\tau_i$, the next $\naux_i$ \readop operations
will return the sequence $\auxt_i$; and the operation \writeop writes
a token to the end of the output stream $R$. The algorithm may perform
at most \paramq total \auxop operations. Additionally, the algorithm
may perform at most \paramy{} \writeop operations between two (not
necessarily consecutive) \readop operations that read consecutive
\maintok{s} $\tau_i, \tau_{i + 1}$. All main, auxiliary, and output
tokens are at most $L$ bits long. Apart from the write-only output
stream $R$ for which \nout may be large, the rest of the space used by
the algorithm must be polynomial in $L$.

The motivation for the restriction that a \specialalgo may access
\auxtok{s} for at most \paramq \maintok{s} is that the algorithms
we wish to simulate use the \maintok{s} as ``summaries'' of their
corresponding sequences of \auxtok{s}. These algorithms operate on
extremely long sequences of input; that is,
$\sum_{1 \le i \le \nin}\naux_i$ is huge. We cannot afford to move all
of the \auxtok{s} around the cluster during a simulation in \congest;
however, most of the time, these algorithms
can obtain all necessary information
about the input from some summarizing function that aggregates the
information in many \auxtok{s} into just one \maintok.
When the algorithm finds that it in fact needs more
granular information, it inspects each of the \auxtok{s} that
correspond to the summarizing \maintok. However, inspecting a sequence
of \auxtok{s} remains an expensive component of the simulation
process, which is why we need to have a good bound \paramq on the
number of times an algorithm may need to access a stream of \auxtok{s}
in order to guarantee a suitable round complexity.

\subsection{The Simulation}

We efficiently simulate \specialalgos when their input is distributed
to the vertices of a communication cluster in a suitable manner.
Specifically, we require that each vertex holds a small, contiguous
part of the input stream. We also require that the vertices hold these
parts in order of their number, i.e., the lowest-numbered vertex holds
some prefix of the input stream, the next-lowest-numbered vertex holds
some prefix of the rest of the stream, etc. The reason for this
requirement is that we simulate the \specialalgo by moving the state
of the algorithm from one vertex to another in the order that they
hold the input. As the degree of vertices in the communication cluster
can be low, we may not have enough bandwidth to distribute information
about which part of the input is held by which vertex; therefore, we
need the input to be given in this predetermined manner. In our
applications, it is not difficult to arrange the input so that this
condition holds.

\begin{definition}[Streaming Input Cluster]%
\label{def:streaming-input-cluster}
    A \emph{streaming input cluster} for a parameter $T_{\max}$ is a
    $(\phi, \delta)$-communication cluster
    $\genericCluster = (\genericClusterNodes, \genericClusterEdges)$
    in $G = (V, E)$ for some values $\phi \in O(\polylog n)$ and
    $\delta$, together with a set of streaming algorithms
    $\algo{A}_1,\dots,\algo{A}_\zeta$ with parameters
    $L,\nin,\nout,\paramq,\paramy$ known to all $v \in \vlist$. The
    token lengths 
    satisfy
    $L = O(\polylog n)$.

    $\genericCluster$ receives the input to
    the algorithms distributed among its vertices so that each vertex
    $v \in \vlist$ holds a contiguous interval $\sigma_{j,v}$ of at
    most $T_{\max}$ \maintok{s} from the input stream $S_j$ to
    algorithm $\algo{A}_j$, and the vertices $\vlist$ are numbered in
    order of the input intervals they hold; that is, concatenating the
    intervals of \maintok{s} $(\sigma_{j,v_1}, \dots, \sigma_{j,v_k})$
    produces exactly $S_j$ for all algorithms $\algo{A}_j$. We call
    this condition \emph{\conscont}. Additionally, for each \maintok
    $\tau_i$ from stream $S_j$ held by a vertex $v$, $v$ also holds
    the associated \auxtok{s}
    $\auxt_{i,1},\auxt_{i,2},\dots,\auxt_{i,\naux_i}$.
\end{definition}

We
introduce a definition that is useful for some components of our
algorithms. Several
procedures presented in this paper involve delegating the
responsibility of collecting information from some set of vertices
$V'$, so that each vertex in a set $\vchain{V}$ is responsible for a
relatively small number of vertices $u \in V'$. We often additionally
require that the vertices $u \in V'$ for which some $v \in \vchain{V}$
is responsible be contiguously numbered. We formalize this notion
here.

\begin{definition}[$(\beta,V')$-Vertex Chain]%
\label{def:vchain}
    Given a graph $G = (V, E)$ and some contiguously-numbered set of
    vertices $V' \subseteq V$, a \emph{$(\beta,V')$-vertex chain} for
    some integer parameter $\beta$ is an ordered set of
    $y = \left\lceil\frac{|V'|}{\beta}\right\rceil$ vertices
    $\vchain{V} \subseteq V$, where the elements of \vchain{V} are
    denoted $\chainmem{V}{1}{},\dots,\chainmem{V}{y}{}$. In addition,
    each $u \in V'$ is \emph{assigned} to some $v \in \vchain{V}$
    according to a many-to-one function $\chainass{V}{u}{}$, which
    satisfies (i) for all $v \in \vchain{V}$, there are at most
    $\beta$ vertices in $\chainres{V}{v}{}$ and they are contiguously
    numbered and known to $v$, and (ii) each $u \in V'$ knows
    $\chainass{V}{u}{}$.
\end{definition}

We now state and prove the conditions for efficient simulation of
\specialalgos in \congest.

\begin{theorem}%
\label{thm:simulate-special}
    Given a streaming input cluster
    $\genericCluster = (\genericClusterNodes, \genericClusterEdges)$
    with $k = |\vlist|$, for any integer parameter
    $1 \le \lambda \le \frac{k}{\zeta}$, all
    $\algo{A}_1,\dots,\algo{A}_\zeta$ can be simulated in parallel in
    \congest in
    $
        \left(
            \frac{T_{\max}}{\delta}(\zeta + \frac{k}{\lambda})
            + (\paramq + 1)(\lambda + \frac{\zeta}{\delta})
        \right)\rcf
    $
    rounds. After the simulation, each token in the output streams
    $R_1,\dots,R_\zeta$ is known to some vertex $v \in \vlist$, and
    each $v \in \vlist$ knows
    $O(\min\{\nout, \paramy T_{\max}\frac{k}{\lambda}\})$ output
    tokens if $\paramq = 0$ and
    $O(\min\{\zeta\nout, \paramy T_{\max}\frac{k}{\lambda}\})$
    output tokens otherwise.
\end{theorem}
\begin{proof}
    The high-level idea of the proof is to simulate each of the
    algorithms $\algo{A}_1,\dots,\algo{A}_\zeta$ over several cluster
    vertices by using the fact that the state
    (which, by definition, has size polynomial in
    $L = O(\polylog n)$)
    of each algorithm
    $\algo{A}_j$ can be sent from one vertex in $\vlist$ to another in
    $\rct$ rounds, using \Cref{thm:cs-routing}.

    The simulation that processes the stream $S_j$ of
    algorithm $\algo{A}_j$ is coordinated by a
    $(\frac{k}{\lambda},\vlist)$-vertex chain $\vchain{\Sim}_j$
    (Definition~\ref{def:vchain}) that contains $\lambda$ vertices; we
    call these chains \emph{simulator chains}, and we make sure that
    they are disjoint. The simulator chain $\vchain{\Sim}_j$ collects
    all \maintok{s} for algorithm $\algo{A}_j$ so that, except for the
    at most \paramq occasions when the algorithm invokes \auxop,
    $\algo{A}_j$ can be simulated entirely by $\vchain{\Sim}_j$.

    The protocol proceeds in several phases, as follows.

    \paragraph{\textbf{Phase 0: Assigning simulator chains in 0 rounds.}}

    All vertices in $\vlist$ locally assign vertices to simulator
    chains $\vchain{\Sim}_j$ for each algorithm $\algo{A}_j$ in an
    identical and predetermined manner. Each vertex is assigned to
    some position in a simulator chain at most once and the simulator
    chains are disjoint. This is possible because the parameter
    $\beta$ of each chain is $\beta = \frac{k}{\lambda}$ and therefore
    contains $\lambda$ vertices, and $\lambda$ satisfies
    $\lambda\zeta \le k$.

    \paragraph{\textbf{Phase 1: Sending stream information in
    \texorpdfstring{%
        $\boldsymbol{
            \frac{T_{\max}}{\delta}
            (\zeta+\frac{k}{\lambda})\rcf
        }$
    }{not too many }rounds.}}

    For each algorithm $\algo{A}_j$, each simulator vertex
    $v = \chainmem{\Sim}{i}{_j}$ receives the \maintok{s} $\tau$ from
    $S_j$ held by all $u \in \chainres{\Sim}{v}{_j}$. Since the number
    of \maintok{s} in $\sigma_{j,u}$ is at most $T_{\max}$, each
    vertex sends at most $\zeta T_{\max}$ tokens during this phase,
    and each simulator vertex receives at most
    $\frac{k}{\lambda}T_{\max}$ tokens. Since each token has length
    $L = O(\polylog(n))$, by Theorem~\ref{thm:cs-routing} all of the
    \maintok{s} $\tau$ can be sent and received in
    $\frac{T_{\max}}{\delta}(\zeta+\frac{k}{\lambda})\rcf$
    rounds.

    At the end of this phase, each \maintok $\tau$ in the stream $S_j$
    is known by some $v \in \vchain{\Sim}_j$. Additionally, because of
    the input contiguity property of
    Definition~\ref{def:streaming-input-cluster}, the stream $S_j$ can
    be partitioned into $\lambda$ contiguous intervals
    $s_{j,1},\dots,s_{j,\lambda}$, where
    $
        s_{j,i} =
            \bigcup_{u\in\chainres{\Sim}{\chainmem{\Sim}{i}{_j}}{_j}}
            \sigma_{j,u}
    $,
    so that each $\chainmem{\Sim}{i}{_j}$ knows the \maintok{s} $\tau$
    in $s_{j,i}$. Because the indices of the endpoints of the
    intervals $\sigma_{j,u}$ may not be known, it is possible that the
    indices of the endpoints of $s_{j,i}$ also are not known even to
    $\chainmem{\Sim}{i}{_j}$~--- however, $\chainmem{\Sim}{i}{_j}$
    knows all $\tau \in s_{j,i}$, knows that this set of tokens is
    guaranteed to be contiguous, and knows the identities of the other
    members of $\vchain{\Sim}_j$.

    \paragraph{\textbf{Phase 2: Simulating the algorithms in
    \texorpdfstring{%
        $\boldsymbol{
            (\paramq+1)(\lambda+\frac{\zeta}{\delta})\rcf
        }$
    }{also not too many }rounds.}}

    The simulator chain $\vchain{\Sim}_j$ can now begin simulating
    $\algo{A}_j$. Initially, all $v = \chainmem{\Sim}{i}{_j}$ for
    $i > 1$ begin in an \emph{inactive} state, in which they wait to
    be \emph{activated} by the previous vertex
    $\chainmem{\Sim}{i - 1}{_j}$. The vertices
    $\chainmem{\Sim}{1}{_j}$ are activated as soon as Phase~1 finishes
    and locally construct the starting state of the algorithm
    $\algo{A}_j$. Upon activation, the simulator vertex
    $\chainmem{\Sim}{i}{_j}$ behaves as follows.
    \begin{itemize}
        \item If the next operation of $\algo{A}_j$ is to perform
            \readop on a \maintok $\tau_l$, $\chainmem{\Sim}{i}{_j}$
            simulates $\algo{A}_j$ using the value of $\tau_l$ if it
            is known (i.e., if $\tau_l \in s_{j,i}$). If the value is
            not known, $\chainmem{\Sim}{i}{_j}$ sends the state of
            $\algo{A}_j$ to the next simulator vertex
            $\chainmem{\Sim}{i + 1}{_j}$ and activates it.
        \item If the next operation of $\algo{A}_j$ is to perform
            \auxop on the \auxtok{s} associated with $\tau_l$,
            $\chainmem{\Sim}{i}{_j}$ simulates $\algo{A}_j$ by sending
            the state of $\algo{A}_j$ to the vertex $v$ that initially
            held $\tau_l$ from $S_j$. The vertex $v$ simulates
            $\algo{A}_j$ until \readop is performed on the next
            \maintok $\tau_{l + 1}$, which it can do since it holds
            all needed \auxtok{s}. Then, $v$ returns the updated state
            of $\algo{A}_j$ to $\chainmem{\Sim}{i}{_j}$. The vertex
            $v$ does not return any output tokens emitted by the
            simulated algorithm; instead, it stores these locally.
        \item If the next operation of $\algo{A}_j$ is a
            \writeop on token $\rho_l$, vertex
            $\chainmem{\Sim}{i}{_j}$ simulates $\algo{A}_j$ by storing
            $\rho_l$ locally.
    \end{itemize}

    During the simulation, most of the congestion occurs around
    vertices that receive many requests to simulate \auxop. The total
    number of \auxop operations is $O(\zeta\paramq)$, thus we hope to
    handle them all in $\frac{\zeta\paramq}{\delta}\rcf$ rounds.
    However, a $v \in \vlist$ can receive up to $\zeta\paramq$ such
    requests, thus if $v$ dispatches the responses in an arbitrary
    order, the request could take $\frac{\zeta\paramq}{\delta}\rcf$
    rounds to answer. This is because a vertex saturated with such
    requests may simulate all $O(\zeta\paramq)$ of a large fraction of
    \auxop operations of other algorithms before processing even the
    first \auxop of some algorithm $\algo{A}_j$. As $\algo{A}_j$
    cannot continue and issue additional \auxop queries until it
    receives the first response, $\algo{A}_j$ may encounter such
    delays during a large fraction of its $\paramq$ requests to
    simulate \auxop operations and expend
    $\omega(\frac{\zeta\paramq^2}{\delta})$ rounds waiting for those
    requests. This would increase the
    $(\paramq + 1)(\lambda + \frac{\zeta}{\delta})\rcf$ term of the
    complexity by a factor of \paramq. Thus we sequence the simulation
    of \auxop operations such that the first \auxop of every algorithm
    is complete before vertices begin the second \auxop of some
    algorithms, ensuring delays cannot not accumulate.

    More precisely, we split Phase~2 into $\paramq + 1$ steps of
    $(\lambda + \frac{\zeta}{\delta})\rcf$ rounds each.
    Within each step, the simulation of each algorithm proceeds until
    the next time \auxop is performed on some sequence of \auxtok{s}
    $\auxt_i$ (or until the end of the algorithm, if there are no more
    invocations of \auxop), then pauses until the beginning of the
    next step. During step $t$ for $1 \le t \le \paramq$, the
    following actions are performed:
    \begin{enumerate}
        \item Any vertex that received a request to simulate
            $\algo{A}_j$ invoking \auxop on $\auxt_i$ performs the
            simulation until \readop is invoked on the next \maintok
            $\tau_{i + 1}$, then sends the updated state of
            $\algo{A}_j$ to the appropriate vertex in
            $\vchain{\Sim}_j$.
        \item Upon receiving the response, each chain
            $\vchain{\Sim}_j$ continues simulating $\algo{A}_j$ until
            the next time \auxop is invoked on a sequence of
            \auxtok{s} $\auxt_{i'}$. Then $\vchain{\Sim}_j$ sends the
            current state of $\algo{A}_j$ to the vertex that holds
            $\auxt_{i'}$.
    \end{enumerate}

    In each step and for each algorithm, $v \in \vlist$ responds to
    and receives at most one request to process a sequence of
    \auxtok{s}, thus the these can be delivered in
    $\frac{\zeta}{\delta}\rcf$ rounds. Additionally, the time required
    to simulate $\algo{A}_j$ before it reads the next \auxtok is
    $\lambda\rcf$, as the only communication required is propagation
    of the state of $\algo{A}_j$ through the at most $\lambda$
    vertices in $\vchain{\Sim}_j$. Phase~2 overall therefore finishes
    in $(\paramq + 1)(\lambda + \frac{\zeta}{\delta})\rcf$ rounds.

    To prove that the distribution of output tokens satisfies the
    guarantees of the theorem, consider the $O(\nout)$ output tokens
    of $\algo{A}_j$. Each $v = \chainmem{\Sim}{i}{_j}$ knows
    $O(\paramy)$ tokens for each $\tau \in s_{j,i}$, of which there
    are
    $
        O(T_{\max}|\chainres{\Sim}{i}{_j}|)
        = O(T_{\max}\frac{k}{\lambda})
    $
    total. Also, if $\paramq \ne 0$, any vertex that processed
    \auxtok{s} holds $O(\paramy)$ tokens for each of the
    $O(T_{\max}\zeta)$ sequences of \auxtok{s} processed. In both
    cases, a vertex knows at most all $O(\nout)$ output tokens for a
    given $\algo{A}_j$. So the number of output tokens held by any
    vertex is $O(\min\{\nout, \paramy T_{\max}\frac{k}{\lambda}\})$ if
    $\paramq = 0$. Otherwise, the number of output tokens known by any
    vertex is
    $O(\min\{\zeta\nout,\paramy T_{\max}(\frac{k}{\lambda}+\zeta)\})$,
    which, since $\zeta$ is bounded by $\frac{k}{\lambda}$, is
    $O(\min\{\zeta\nout,\paramy T_{\max}\frac{k}{\lambda}\})$.
\end{proof}

\section{Partition Trees}
\label{sec:partitionTrees}
We describe two classes of partition trees. Due to the low bandwidth
of \congest and the need to incorporate information from vertices
outside high-conductance expanders, our partition trees require
stronger balancing properties than those initially presented in
\cite{CFG+21} for \clique. We demonstrate that, within a suitable
communication cluster, the partition tree needed to list all instances
of $K_p$, for $p \ge 3$, can be constructed in $n^{1 - 2/p \rc}$
rounds.

\begin{definition}[\hspace{-1pt}\cite{CFG+21}, $p$-partition tree]%
\label{def:ptree}
    Let $G = (V, E)$ be a graph with $n = |V|$, $m = |E|$, and let
    $p \leq \log n$. A \emph{$p$-partition tree} $T = T_{G,p}$ is a
    tree of $p$ layers (depth $p - 1$), where each non-leaf node has
    at most $x = n^{1/p}$ children. Each tree node is associated with
    a partition of $V$ consisting of at most $x$ parts.

    The partition associated with the root of $T$ is denoted
    $U_{\emptyset}$. Given a node with partition
    $U_{\left( \ell_1,\dots,\ell_{i - 1} \right)}$, the partition
    associated with its $j$th child, for $0 \leq j \leq x - 1$, is
    $U_{\left( \ell_1,\dots,\ell_{i - 1}, j \right)}$.

    The parts of each $U_{\left( \ell_1,\dots,\ell_i \right)}$ are
    denoted $U_{\left( \ell_1,\dots,\ell_i \right),j}$, for
    $0 \leq j \leq x - 1$. For each $0 \leq j \leq x - 1$, also denote
    $
        {\tt parent} (
                U_{\left( \ell_1,\dots,\ell_{i - 1},\ell_i \right),j}
            ) =
        U_{\left( \ell_1,\dots,\ell_{i - 1} \right), \ell_i}
    $
    and recursively define the set of ancestor parts for
    $S = \left( \ell_1,\dots,\ell_{i - 1},\ell_i \right)$
    \[
        {\tt anc}(U_{S,i}) =
        \begin{cases}
            {\tt anc}({\tt parent}(U_{S,i}))
            \cup \{U_{S,i}\}
            &
            S \ne \emptyset
            \\
            \{U_{{\emptyset},i}\}
            &
            S = \emptyset
        \end{cases}
    \]
\end{definition}

The following distributes the work of detecting all instances of a
subgraph $H$ by assigning the task of learning all edges between parts
in ${\tt anc}(U_{S,i})$ for each leaf part $U_{S,i}$ to some vertex.
The idea of the proof is to consider an instance $H'$ of $H$ in $G$
with vertices $v_0,v_1,\dots,v_{p - 1}$ and trace a path from the root
of a $p$-partition tree $T$ to the leaf layer, choosing first the part
that contains $v_0$, then the part that contains $v_1$, etc. As each
vertex is in a distinct part, each edge of $H'$ must be be present
between some pair of distinct parts. We refer the reader to the proof
of~\cite[Theorem~11]{CFG+21}.

\begin{theorem}%
\label{thm:cfg-tree-can-list}
    For any instance $H'$ of a $p$-vertex subgraph $H$ in $G = (V, E)$
    and any $p$-partition tree $T$ of $G$, there exists a leaf part
    $U_{S,i}$ in $T$ for which all edges of $H'$ are contained in the
    union
    $
        \bigcup_{U \in {\tt anc}(U_{S,i})}
        \bigcup_{W \in {\tt anc}(U_{S,i}) : W \ne U} E(U,W)
    $.
\end{theorem}


\subsection{Partition Trees for\texorpdfstring{
    $\boldsymbol{K_3}$
}{
    K\_3
}Listing}

Theorem~\ref{thm:cfg-tree-can-list} shows that the leaves of a
$p$-partition tree can be used to distribute the work required for
subgraph listing. In~\cite{CFG+21}, $p$-partition trees are further
constrained to ensure that the number of edges from each part
$U_{S,j}$ to its ancestors ${\tt anc}(U_{s,j})$ is balanced to within
some small, additive error term of at most $n$. However, in \congest,
we are forced to reduce the permissible error to $n^{2/3}$. Following
is our definition of partition tress with the stronger constraints.

\begin{definition}[$H$-partition Tree]%
\label{def:htree-congest}
    Let $G' = (V', E')$ be a subgraph of $G = (V, E)$ with $n = |V|$,
    $k = |V'|$, $m = |E'|$, and let $H = (V_H, E_H)$ be a graph with
    $p = |V_H| \leq \log n$ vertices
    $\left\{ z_0,\dots,z_{p - 1} \right\}$, and denote
    $
        d_i = \left|\left\{
            \left\{ z_i, z_t \right\} \in E_H \; | \; t < i
        \right\}\right|
    $
    for each $0 \leq i \leq p - 1$, $x = k^{1/p}$ and
    $\tilde{m} = \max\left\{ m, kx \right\}$. \emph{An $H$-partition
    tree} $T = T_{G',H}$ is a $p$-partition tree with the following
    additional constraints, for some constants $c_1, c_2, c_3$:
    \begin{enumerate}
        \item \consdeg: For every part $U_{S,j}$, it holds that
            $|E(U_{S,j},V')| \leq c_1\tilde{m}/x$.
        \item \consupdeg: For every part $U_{S,j}$, and for all of its
            ancestor parts $W \in {\tt anc}(U_{S,j})$,\\
            $
                \sum_{W \in {\tt anc}(U_{S,j})\setminus\{U_{S,j}\}}
                    \left| E( U_{S,j}, W ) \right| \leq
                    c_2d_i\tilde{m}/x^2 + c_3pk/x
            $.
        \item \conssz: For every part $U_{S,j}$, it holds that
            $|U_{S,j}| \leq c_3k/x$.
    \end{enumerate}
\end{definition}

We show that the simulation techniques of
Section~\ref{sec:special-algo} can efficiently construct the more
strongly constrained $p$-partition tree for $K_3$ listing in \congest.
We describe the properties of a cluster that allow us to construct a
$K_3$-partition tree on it and formally state the guarantees of the
$K_3$-partition-tree construction. We use
$(\phi, \delta)$-communication clusters
(Definition~\ref{def:com-cluster}) with suitable values of $\phi$ and
$\delta$. As before, for $C$ and its corresponding $\vlist$, we denote
$n = |V|$, $k = |\vlist|$, $K = |\genericClusterNodes|$. We also let
$\mu$ be the average communication degree of vertices in $\vlist$ and
$\vhd$ be the set of $v \in \vlist$ with at least half of the average
communication degree.

\begin{definition}[$K_3$-Compatible Cluster]%
\label{def:k3-cluster}
    A \emph{$K_3$-compatible cluster} is a
    $(\phi, \delta)$-communication cluster
    $\genericCluster = (\genericClusterNodes, \genericClusterEdges)$
    in $G = (V, E)$ for $\phi = O(\polylog n)$ and
    $\delta = K^{1/3}$.
\end{definition}

\begin{theorem}%
\label{thm:congest-tree-algo-tri}
    Given a $K_3$-compatible cluster
    $\genericCluster = (\genericClusterNodes, \genericClusterEdges)$,
    there exists a deterministic \congest algorithm on
    $\genericCluster$ that, in $k^{1/3}\rcf$
    rounds, constructs a $K_3$-partition tree $T$ of $C[\vlist]$,
    such that:
    \begin{itemize}
        \item The root and first layer of $T$ are known to all
            vertices $\vlist$.
        \item Each vertex $v \in \vhd$ knows
            $O(\frac{1}{\mu}\comdeg v)$ parts of
            the leaf layer of $T$.
        \item Each part of the leaf layer is known to some vertex
            $v \in \vhd$.
    \end{itemize}
\end{theorem}

We begin by proving that there is a \specialalgo for constructing one
layer of an $H$-partition tree. This will then allow us to apply the
algorithm of Theorem~\ref{thm:simulate-special} for simulating
\specialalgos to efficiently construct one layer of a $K_3$-partition
tree of $C[\vlist]$.

\begin{lemma}%
\label{lem:k3-htree-streaming}
    There exist constants $c_1,c_2,c_3$ such that given a graph
    $G = (V, E)$ with $|V| = n$, subgraphs $G' = (V', E')$ and
    $H = (V_H, E_H)$ with $|V'| = K$ and $|V_H| = p < \log n$,
    $0 \le i < p$ levels of an $H$-partition tree $T_{G',H}$ for
    $c_1,c_2,c_3$, and a part $U$ in the $i$-th level of $T$, there
    exists a \specialalgo with parameters $L = O(\polylog n)$,
    $\nin = \tilde{O}(n)$, $\nout = O(n^{1/p})$, $\paramq = 0$,
    $\paramy = \nout$ that reads a stream $S$ and produces a stream
    $R$, where:
    \begin{itemize}
        \item The elements of $R$ are the endpoints of intervals of
            vertex numbers that define a valid child partition of $U$
            in some $H$-partition tree for $c_1,c_2,c_3$.
        \item Stream $S$ contains $|E(v, U')|$ for each part
            $U' \in {\tt anc}(U)$ for each $v \in V'$ in order of
            increasing vertex number.
    \end{itemize}
\end{lemma}

\begin{proof}
    We process vertices one at a time, maintaining counters for each
    of the balancing constraints \consdeg, \consupdeg, and \conssz in
    the $H$-partition tree definition. We greedily add vertices to the
    current part until it is no longer possible to do so without
    causing a counter to exceed its maximum value, at which point we
    start a new part.

    We show that this counter-based algorithm is a
    \specialalgo with the desired parameters. Upon initialization, the
    algorithm records in its state that the first interval begins with
    vertex 1. The algorithm then maintains one counter for each of the
    3 constraints \consdeg, \consupdeg, and \conssz in
    Definition~\ref{def:htree-congest} on the size and number of edges
    in an $H$-partition tree part. Each vertex is represented in the
    input stream $S$ by a constant number of tokens that give its
    number, its total degree, and its degree into each part in
    ${\tt anc}(U)$; when the algorithm processes a vertex, it reads
    these tokens and adds them to the corresponding counters. If upon
    processing vertex $i$ any counter exceeds the maximum value, the
    algorithm outputs that vertex $i - 1$ is the last vertex of the
    previous interval and stores in its state that vertex $i$ is the
    first vertex of the next one; the algorithm then resets the values
    of all counters to zero and adds the corresponding values received
    from vertex $i$.

    The entire state of the algorithm is maintained by $\tilde{O}(1)$
    counters, each of which requires $O(\log n)$ bits, so
    $L = O(\polylog n)$. The algorithm only uses the $O(N)$ main
    tokens in the stream, so it never performs \auxop; therefore,
    $\paramq = 0$. The total number of \writeop operations performed
    by the algorithm is $\nout$, so the number of \writeop operations
    between any two \readop operations is bounded by $\paramy = \nout$
    as well. The number of input tokens is linear in the size of the
    graph, so $\nin = O(n)$. The number of output tokens is linear in
    the number of child parts of a given part, so that by
    Definition~\ref{def:htree-congest} we have $\nout = O(n^{1/p})$.

    We now show that there exist constants $c_1,c_2,c_3$ for which
    this algorithm is guaranteed to produce a valid partition in an
    $H$-partition tree, namely, that each part of the partition
    satisfies the constraints \consdeg, \consupdeg, and \conssz, and
    that the number of parts in this partition is at most
    $x = k^{1/p}$. Since the counters guarantee that each of the
    properties of an individual part are satisfied, all that remains
    is to show that the number of parts in any partition is at most
    $x = k^{1/p}$. First, in any partition, each vertex can be
    responsible for the overflow of any counter at most once, and in
    particular it is always possible to start a new partition without
    exceeding the maximum value of any counter. For any $c_1 \geq 1$
    and for all $v$, it holds that
    $c_1\tilde{m}/x \geq c_1(kx)/x \geq k \geq \deg_{V'} v$.
    Additionally, since we enforce that each part $U_{S,i}$ of each
    partition satisfies $\left|U_{S,i}\right| \leq c_3k/x$, we have
    that
    \[
        \sum_{W \in {\tt anc}(U_{S,i})}
            \left|E( v, W )\right|
            \leq c_3|S|k/x \leq c_3pk/x.
    \]
    We can bound the number of parts in each partition by bounding the
    number of times these counters exceed their maximum values. A part
    that ends due to counter \consdeg has at least
    $
        c_1\tilde{m}/x - k
            \geq (c_1 - 1)\tilde{m}/x
            \geq (c_1 - 1)m/x
    $
    edges out of $2m$ total, so that there are at most $2x/(c_1 - 1)$
    such parts. Likewise, a part that ends due to counter \consupdeg
    has, for the given set of parts in the ancestors of $U$, at least
    $c_2|S|\tilde{m}/x^2$ edges into them, out of at most
    $c_1\tilde{m}|S|/x$ total, so there are at most $c_1x/c_2$ such
    parts. Finally, the number of times that counter \conssz starts a
    new part is at most $x/c_3$. The maximum number of parts in any
    partition is therefore $2x/(c_1 - 1) + c_1x/x_2 + x/c_3 + 1$,
    which for $c_1 = 9, c_2 = 36, c_3 = 4$ is at most $x$.
\end{proof}

We can now use the algorithm for the simulation of \specialalgos
described in \Cref{sec:special-algo} to build a single layer of a
$K_3$-partition tree in $k^{1/3}\rcf$ rounds. We state this
formally in the following.

\begin{lemma}%
\label{lem:k3-htree-one-layer}
    There exist constants $c_1,c_2,c_3$ such that, given a
    $(\phi, \delta)$-communication cluster
    $\genericCluster = (\genericClusterNodes, \genericClusterEdges)$
    in graph $G = (V, E)$ for $\phi = O(\polylog n)$ and
    $\delta = K^{1/3}$, where each $v \in \vlist$ knows the first
    $0 \le i < 3$ layers of some $K_3$-partition tree $T_{C,K_3}$ with
    constants $c_1,c_2,c_3$ in \congest, there exists a deterministic
    \congest algorithm that finishes in $k^{1/3}\rcf$
    rounds and constructs a next layer for the tree $T_{C,K_3}$ so
    that each endpoint of a part of the next layer is known to some
    vertex $v \in \vlist$ and each $v \in \vlist$ knows $O(k^{1/3})$
    endpoints of parts.
\end{lemma}

\begin{proof}
    We invoke \Cref{thm:simulate-special} on the streaming algorithm
    implied by Lemma~\ref{lem:k3-htree-streaming}. We execute several
    instances of this \specialalgo in parallel, one for each part of
    the lowest completed level of the $K_3$-partition tree $T$. Since
    each vertex $v \in \vlist$ knows all prior layers of the tree
    $T_{C,K_3}$ and its own edges, it holds the $O(1)$ tokens that
    describe it for each of the $O(k^{2/3})$ parts in layer
    $0 \le i < 3$. Applying \Cref{thm:simulate-special} with
    parameters $T_{\max} = O(1), \lambda = k^{1/3}$ gives that all
    $O(k^{2/3})$ algorithms finish in parallel in
    $
            (\frac{1}{K^{1/3}}(k^{2/3} + \frac{k}{k^{1/3}})
            + (k^{1/3} + \frac{k^{2/3}}{K^{1/3}}))
            \rcf
        =
        k^{1/3}\rcf
    $
    rounds as desired. Together, the output streams of all of the
    executed algorithms contain all parts in the next layer of $T$.
    Since $\paramq = 0$, each $v \in \vlist$ knows
    $O(\nout) = O(k^{1/3})$ output tokens.
\end{proof}

\paragraph{\textbf{Load balancing the output.}} The distribution of
the output tokens of the \specialalgo as simulated with
Theorem~\ref{thm:simulate-special} is very loosely constrained, so we
show load balancing tools in order to redistribute it. After applying
Lemma~\ref{lem:k3-htree-one-layer} to construct a new layer of a
$K_3$-partition tree, each part of the new layer is known only to one
vertex, and the only guarantee about how many parts each vertex knows
is that it is $O(k^{1/3})$. In order to apply
Lemma~\ref{lem:k3-htree-one-layer} again to construct the next layer,
and in order to satisfy the guarantees of
Theorem~\ref{thm:congest-tree-algo-tri}, we wish to distribute the
root and middle layers of the $K_3$-partition tree so that they are
known by all vertices in $\vlist$. We show how to do this in
Lemma~\ref{lem:k3-amplify-distribute}. Likewise, in order to satisfy
the guarantees of Theorem~\ref{thm:congest-tree-algo-tri}, we wish to
distribute the parts of the leaf layer so that each $v \in \vhd$ knows
$O(\frac{1}{\mu}\comdeg v)$ parts. We accomplish this in
Lemma~\ref{lem:k3-chain-distribute}.

\begin{lemma}%
\label{lem:k3-amplify-distribute}
    Given a $(\phi, \delta)$-communication cluster
    $\genericCluster = (\genericClusterNodes, \genericClusterEdges)$
    in graph $G = (V, E)$ for $\phi = O(\polylog n)$,
    $\delta = K^{1/3}$, $K = |\vcluster|$, and $k = |\vlist|$, and
    $O(k^{2/3})$ numbered messages that must be learned by all
    $v \in \vlist$, where initially each message is known by a unique
    $v \in \vlist$ and each $v \in \vlist$ knows $O(k^{1/3})$
    messages, there is a deterministic \congest algorithm that
    finishes in $k^{1/3}\rcf$ rounds and distributes all
    $O(k^{2/3})$ messages to all vertices in $\vlist$.
\end{lemma}

\begin{proof}
    We begin by deterministically assigning a
    $(k^{2/3},\vlist)$-vertex chain $\vchain{A}_j$ for each message
    $m_j$ in a way that each vertex $v \in \vlist$ can compute
    locally. We call these chains \emph{amplifier chains}. All of the
    elements $\chainmem{A}{i}{_j}$ of the amplifier chains are in
    $\vlist$, and each $v \in \vlist$ is in $O(1)$ chains. This is
    possible because the number of elements in each amplifier chain is
    $O(|\vlist| / k^{2/3}) = O(k^{1/3})$ and there are $O(k^{2/3})$
    such chains.

    The messages are distributed in two phases, each of which finishes
    in $k^{1/3}\rcf$ rounds:
    \begin{enumerate}
        \item Each vertex $u \in \vlist$, for each message $m_j$
            initially held by $u$, sends $m_j$ to all
            $v \in \vchain{A}_j$.
        \item For each message $m_j$, each vertex $v \in \vchain{A}_j$
            sends $m_j$ to all $u \in \chainres{A}{v}{_j}$, where
            $\chainres{A}{v}{_j}$ is the set of $u \in \vlist$
            assigned to $v$ in the vertex chain.
    \end{enumerate}

    In Phase~1, each vertex initially holds $O(k^{1/3})$ messages and
    sends each message to each of the $O(k^{1/3})$ vertices in
    $\vchain{A}_j$. Each vertex receives 1 message for each of the
    $O(1)$ amplifier chains of which it is a part. Phase~1 therefore
    finishes in $k^{1/3}\rcf$ rounds.

    In Phase~2, each vertex sends the message $m_j$ to $O(k^{2/3})$
    vertices for each of the $O(1)$ amplifier chains $\vchain{A}_j$ of
    which it is a part. Each vertex receives all $O(k^{2/3})$ messages
    $m_j$. Phase~2 therefore finishes in $k^{1/3}\rcf$
    rounds as well.
\end{proof}

In the following lemma, we show that a $K_3$-compatible cluster can
distribute $O(k)$ messages so that each $v \in \vhd$ knows
$O(\frac{1}{\mu}\comdeg v)$ messages. We present a \specialalgo that
has as its input a list of the degrees of vertices $v \in \vlist$. For
each $v \in \vlist$, the algorithm outputs an interval of messages
that $v$ is responsible for learning. We simulate the algorithm by
Theorem~\ref{thm:simulate-special} and redistribute the intervals so
that each $v$ knows the numbers of the messages it is responsible for
learning. For each such message, $v$ can ask the vertex $u$ that holds
that message for the contents of the message.

\begin{lemma}%
\label{lem:k3-chain-distribute}
    Given a $(\phi, \delta)$-communication cluster
    $\genericCluster = (\genericClusterNodes, \genericClusterEdges)$
    in graph $G = (V, E)$ for $\phi = O(\polylog n)$,
    $\delta = K^{1/3}$, $K = |\vcluster|$, $k = |\vlist|$, and $O(k)$
    numbered messages distributed so that each message is initially
    known by exactly one $v \in \vlist$ and each $v \in \vlist$
    initially knows $O(k^{1/3})$ messages, there is a deterministic
    \congest algorithm that finishes in $k^{1/3}\rcf$ and
    redistributes the messages so that each $v \in \vhd$ knows
    $O(\frac{1}{\mu}\comdeg v)$ messages.
\end{lemma}

\begin{proof}
    By Theorem~\ref{thm:diameter-conductance}, the diameter of $C$ is
    $O(\polylog n)$, thus $\vlist$ can compute the total
    communication degree $m = |E(\vlist, \vcluster)|$, the average
    communication degree $\mu$, and the total number of messages $M$
    and distribute the quantities $m$, $\mu$, and $M$ so that they are
    known by all $\vlist$ in $\tilde O(1)$ rounds. 
    
    Let $c = \left\lceil M/k \right\rceil$. Since $M = O(k)$, it holds
    that $c = O(1)$. We begin by deterministically assigning at most
    $c$ messages to each $v \in \vlist$ in a way that can be locally
    computed by all vertices. In particular, message $m_j$ is assigned
    to the vertex with number $\lceil j/c \rceil$. Then each vertex
    sends each message $m_j$ it originally holds to vertex
    $\lceil j/c \rceil$. Each vertex sends $O(k^{1/3})$ messages and
    receives $O(1)$ messages, so this process requires $\rct$ rounds.

    We present a simple \specialalgo (see
    Algorithm~\ref{alg:k3-chain-distribute}) that, given an input
    stream $S$ containing the degrees of all $v \in \vlist$, produces
    an output stream $R$, where each token contains the endpoints of a
    range of messages to be delivered to $v$. All messages are
    assigned to some $v \in \vhd$ in this way and each $v \in \vhd$
    receives $O(\frac{1}{\mu}\comdeg v)$ messages. Each input token is
    a tuple $(v, \comdeg v)$ and each output token is a tuple
    $(v, I)$, where $I$ is $\emptyset$ or some $[a, b]$.

    \begin{algorithm}
        \caption{Balance messages by communication degree}%
        \label{alg:k3-chain-distribute}
        \begin{algorithmic}[1]
            \State $\textit{leaf} \gets 0$
            \For {$\tau_i \in S$}
                \State \readop $\tau_i = (v_i, \comdeg {v_i})$
                \If {$\comdeg {v_i} < \mu/2$}
                    \State \writeop $(v, \emptyset)$
                \Else
                    \State $l \gets 2\lceil M\comdeg {v_i} / m\rceil$
                    \State \writeop
                        $(v, [\textit{leaf} + 1, \textit{leaf} + l])$
                    \State $\textit{leaf} \gets \textit{leaf} + l$
                \EndIf
            \EndFor
        \end{algorithmic}
    \end{algorithm}

    The algorithm allocates at most $2\lceil M\comdeg v / m\rceil$
    messages to each $v \in \vhd$. Since $M = O(k)$ and
    $m = E(\vlist, \vcluster)$, it holds that $M/m = O(1/\mu)$, so
    that the number of messages delegated to each $v \in \vhd$ is
    $O(\frac{1}{\mu}\comdeg v)$. All messages are allocated in this
    way, as
    \[
        \sum_{v \in \vhd} 2\comdeg v \ge k\mu.
    \]
    Since $m = k\mu$, multiplying both sides by $M$ and rearranging
    gives
    \[
        \sum_{v \in \vhd} 2\frac{\comdeg v \cdot M}{m} \ge M.
    \]
    The total communication degree of $\vhd$ is at least half $k\mu$
    as $E(\vhd, \vcluster) \ge E(\vlist \setminus \vhd, \vcluster)$.
    We can show this by contradiction. Let
    $\vlistld = \vlist \setminus \vhd$ be the set of low degree
    vertices and suppose instead that
    $E(\vhd, \vcluster) < E(\vlistld, \vcluster)$. Then
    \[
        |E(\vlist, \vcluster)|
        = |E(\vhd, \vcluster)| + |E(\vlistld, \vcluster)|
        < 2|E(\vlistld, \vcluster)|
        < 2\sum_{v \in \vlistld}\frac{1}{2}\mu
        = |\vlistld|\mu
        \le |\vlist|\mu
        = |E(\vlist, \vcluster)|,
    \]
    which is a contradiction. So we must have
    $E(\vhd, \vcluster) \ge E(\vlistld, \vcluster)$, from which
    $E(\vhd, \vcluster) \ge \frac{1}{2}k\mu$ as desired.

    We simulate the \specialalgo presented by invoking
    \Cref{thm:simulate-special}. First, note that the token length is
    $L = O(\log n)$, as each token contains $O(1)$ numbers, each of
    which is $O(n)$. Additionally, the state of the algorithm consists
    of the numbers $l$ and \textit{leaf} and one input token, and so
    is polylogarithmic in $L$. Since the input and output streams both
    contain one token for each $v \in \vlist$, we get that
    $\nin, \nout = O(k)$. Each vertex initially holds its own input
    token, so $T_{\max} = 1$. The algorithm never performs an \auxop
    operation, so $\paramq = 0$. Finally, $\paramy = 1$, because the
    algorithm outputs exactly one token between \readop operations.
    Thus, using $\lambda = k^{1/3}$, the simulation finishes in
    $
            (\frac{1}{K^{1/3}}(1 + \frac{k}{k^{1/3}})
            + (k^{1/3} + \frac{1}{K^{1/3}}))
            \rcf
            = k^{1/3}\rcf
    $
    rounds. At the end of the simulation, each $v \in \vlist$ knows
    $O(\paramy T_{\max}\frac{k}{\lambda}) = O(k^{2/3})$ output tokens.

    We can now redistribute the messages. First,
    each $v \in \vlist$ sends each of the output tokens it holds to
    the vertex whose assignment is contained in the token. Each vertex
    sends $O(k^{2/3})$ messages and receives $O(1)$ messages, and as
    $\delta = K^{1/3} \ge k^{1/3}$, this can be done in
    $k^{1/3}\rcf$ rounds. Next, each $v \in \vlist$, upon
    receiving the endpoints of its interval $[a, b]$ of message
    numbers, for each $j \in [a, b]$ sends a request to vertex
    $\lceil j/c \rceil$ to obtain message $m_j$. 
    In this step, each
    $v \in \vlist$ sends $\tilde{O}(\frac{1}{\mu}\comdeg v)$ messages
    and receives $O(1)$ messages, so this can be done in
    $\rct$ rounds. The entire process takes
    $k^{1/3}\rcf$ rounds, as desired.
\end{proof}

We are now ready to prove Theorem~\ref{thm:congest-tree-algo-tri}.

\begin{proof}[Proof of Theorem~\ref{thm:congest-tree-algo-tri}]
    Apply Lemma~\ref{lem:k3-htree-one-layer} to build
    the root partition of a $K_3$-partition tree $T$ and
    Lemma~\ref{lem:k3-amplify-distribute} to make it known to all
    $\vlist$ in $k^{1/3}\rcf$ rounds. Then, again invoke
    Lemmas~\ref{lem:k3-htree-one-layer}
    and~\ref{lem:k3-amplify-distribute} to build and distribute the
    middle layer of $T$ in $k^{1/3}\rcf$ rounds. Finally,
    Lemma~\ref{lem:k3-htree-one-layer} builds the leaf
    layer of $T$ in $k^{1/3}\rcf$ rounds and
    Lemma~\ref{lem:k3-chain-distribute} distributes it according to
    the guarantees of Theorem~\ref{thm:congest-tree-algo-tri}.
\end{proof}

\subsection{Partition Trees for\texorpdfstring{
    $\boldsymbol{K_p}$
}{
    K\_p
}Listing,\texorpdfstring{
    $\boldsymbol{p > 3}$
}{
    p>3
}}

For $p > 3$, we are less constrained by the per-round bandwidth of the
cluster, as our target listing time is higher. Thus, we can adopt an
error term of $n$ on the edge-balancing upper bounds. However, because
in $K_p$-listing algorithms for $p > 3$, a cluster $C$ can be
responsible for listing cliques that have some vertices in
$V \setminus \vcluster$, we must adapt the partition trees to
simultaneously balance the distribution of three types of edges: edges
that lie entirely inside $C$; edges that lie entirely outside $C$; and
edges that cross the boundary of $C$ (with one vertex in $C$ and the
other outside). Definition~\ref{def:split-graph} describes these three
edge sets, and Definition~\ref{def:split-tree} describes a new class
of partition tree with these more complex balancing properties.

\begin{definition}[Split Graph]%
\label{def:split-graph}
    A \emph{split graph} is a graph $G = (V, E)$, together with
    disjoint sets $V_1$, $V_2$, $E_1$, $E_2$, $E_{12}$ such that
    $V = V_1 \cup V_2$, $E = E_1 \cup E_2 \cup E_{12}$,
    $E_1 \subseteq V_1 \times V_1$,
    $E_2 \subseteq V_2 \times V_2$, and
    $E_{12} \subseteq V_1 \times V_2$.
    Denote
    $n = |V|$,
    $k = |V_1|$,
    $m_1 = |E_1|$,
    $m_2 = |E_2|$, and
    $m_{12} = |E_{12}|$.
\end{definition}

We construct a partition tree over a split graph, where each vertex is
a partition of either $V_1$ or $V_2$. The constraints \consdegbb,
\consupdegbb, \consdegaa, and \consupdegaa in the definition below
directly correspond to the constraints \consdeg and \consupdeg from
Definition~\ref{def:htree-congest} with respect to vertices from the
same set as a given partition, while the constraints \consdegba and
\consupdegab provide additional balancing for partitions of different
vertex sets.

\begin{definition}[$(p',p)$-Split $K_p$-Tree]%
\label{def:split-tree}
    Let $G$ be a split graph, and let $1 < a \leq b < n$ and
    $1 \leq p' \leq p \leq \log n$ be some integer parameters. 
    Define
    $\tilde{m}_1 = \max\{m_1, ka\}$,
    $\tilde{m}_2 = \max\{m_2, nb\}$, and
    $\tilde{m}_{12} = \max\{m_{12}, na\}$.

    A $(p',p)$-split partition tree $T$ is a tree of $p$ layers (depth
    $p - 1$), where each non-leaf node in the first $p - p'$ layers
    has at most $b$ children and is associated with a partition of
    $V_2$ into at most $b$ parts, and each non-leaf node in the bottom
    $p'$ layers has at most $a$ children and is associated with a
    partition of $V_1$ into at most $a$ parts. The partitions
    $U_S$, the part ${\tt parent}(U_{S,j})$ and set of parts
    ${\tt anc}(U_{S,j})$ are defined inductively as in
    Definition~\ref{def:ptree}.

    A $(p',p)$-split $K_p$-partition tree is a $(p',p)$-split
    partition tree that satisfies the following additional constraints
    for $\pi = p - p'$ and some positive constants $c_1,c_2$.
    \begin{enumerate}
        \item \consdegbb: For every part $U_{S,j}$, $|S| < \pi$,
            it holds that $|E(U_{S,j},V_2)| \leq c_1m_2/b + n$.
        \item \consupdegbb: For every part $U_{S,j}$, $|S| < \pi$,
            and for all of its ancestor parts
            $U_{S',j'} \in {\tt anc}(U_{S,j})$, it holds that
            $
                \sum_{
                    U_{S',j'} \in
                    {\tt anc}(U_{S,j}) \setminus \{U_{S,j}\}
                }
                    |E(U_{S,j}, U_{S',j'})|
                \leq c_2|S|\tilde{m}_2/b^2 + n
            $.
        \item \consdegba: For every part $U_{S,j}$, $|S| < \pi$,
            it holds that $|E(U_{S_j},V_1)| \leq c_1m_{12}/b + n$.
        \item \consdegaa: For every part $U_{S,j}$, $|S| \ge \pi$,
            it holds that $|E(U_{S,j},V_1)| \leq c_1m_1/a + k$.
        \item \consupdegaa: For every part $U_{S,j}$, $|S| \ge \pi$,
            and for all of its ancestor parts
            $U_{S',j'} \in {\tt anc}(U_{S,j}), |S'| \ge \pi$, it
            holds that
            $
                \sum_{
                    U_{S',j'} \in
                        {\tt anc}(U_{S,j}) \setminus \{U_{S,j}\}
                    : |S'| \ge \pi
                }
                    |E(U_{S,j}, U_{S',j'})|
                \leq c_2(|S| - \pi)\tilde{m}_1/a^2 + k
            $.
        \item \consupdegab: For every part $U_{S,j}$, $|S| \ge \pi$,
            and for all of its ancestor parts
            $U_{S',j'} \in {\tt anc}(U_{S,j}), |S'| < \pi$, it
            holds that
            $
                \sum_{
                    U_{S',j'} \in {\tt anc}(U_{S,j})
                    : |S'| < \pi
                }
                    |E(U_{S,j}, U_{S',j'})|
                \leq c_2\pi\tilde{m}_{12}/(ab) + n
            $.
    \end{enumerate}
\end{definition}

The following is an analogue of Theorem~\ref{thm:cfg-tree-can-list}
for $(p,p')$-split partition trees, and shows that such trees can be
used to load balance the listing of all instances of a subgraph $H$ in
a split graph $G$ with $p' \ge 2$ vertices in $V_1$. Consider an
instance $H'$ of $H$ in $G$ with vertices $v_0,v_1,\dots,v_{p - 1}$,
where the first $p - p'$ vertices are in $V_2$ and the remaining $p'$
are in $V_1$, and trace a path from the root of a $(p,p')$-split
partition tree $T$ to the leaf layer, choosing first the part that
contains $v_0$, then the part that contains $v_1$, etc. Since each
vertex will be in a distinct part, each edge of $H'$ must be be
present between some pair of distinct parts. As each vertex is in a
distinct part, each edge of $H'$ must be be present between some pair
of distinct parts.

\begin{theorem}%
\label{thm:split-tree-can-list}
    For any instance $K_p'$ of a $K_p$ in a split graph $G$, where
    $K_p'$ has at least $p' \ge 2$ vertices in $V_1$, and any
    $(p,p')$-split partition tree $T$ of $G$, there exists a leaf part
    $U_{S,i}$ in $T$ for which all edges of $K_p'$ are contained in
    \[
        \bigcup_{U \in {\tt anc}(U_{S,i})}
        \bigcup_{W \in {\tt anc}(U_{S,i}) : W \ne U}
        E(U,W).
    \]
\end{theorem}

To list within a cluster, we define a concept that generalizes that of
a $K_3$-compatible cluster. However, a cluster $C$ can be responsible
for listing cliques that do not lie entirely inside $C$, and so must
receive some input about edges not incident to vertices in $C$.
Therefore, in addition to the constraint in the definition
of a $K_3$-compatible cluster that all elements of $\vlist$ have
communication degree at least $n^{1 - 2/p}$,
Definition~\ref{def:input-cluster} details the necessary input and
gives some constraints about its distribution. These constraints are
easy to satisfy, but do not easily permit us to construct the desired
$(p,p')$-split $K_p$-partition tree. Thus,
Definition~\ref{def:centralized-input-cluster} gives stricter
constraints on the distribution of input within the cluster. We
show that a cluster can rearrange the input given in
Definition~\ref{def:input-cluster} to satisfy the conditions of
Definition~\ref{def:centralized-input-cluster}, and then construct a
$(p,p')$-split $K_p$-partition tree using
Theorem~\ref{thm:simulate-special}.

\begin{definition}[$K_p$-Compatible Cluster]%
\label{def:input-cluster}
    A \emph{$K_p$-compatible cluster} for $p > 3$ is a
    $(\phi, \delta)$-communication cluster
    $\genericCluster = (\genericClusterNodes, \genericClusterEdges)$
    in $G = (V, E)$ for $\phi = O(\polylog n)$ and
    $\delta = n^{1 - 2/p}$. The cluster holds two sets
    of directed edges
    $\ebar \subseteq E(V \setminus \vlist, \vlist)$ and
    $\ep \subseteq E(V \setminus \vlist, V \setminus \vlist)$.
    Define the
    \emph{input degree} $\sentdeg v$ of $v \in V$ to be the
    number of edges in $\ebar \cup \ep$ with tail $v$,
    $\sentdeg v = |\{ u \in V : (v, u) \in \ebar \cup \ep\}|$.
    $\ebar$ and $\ep$ are distributed to $\vlist$ as follows:
    \begin{itemize}
        \item Any vertex $v \in \vlist$ is incident to
            $\tilde{O}(\comdeg v)$ edges from $\ebar$ and knows
            all such edges $(u, v)$.
        \item Each vertex $u \in \vlist$ receives
            $\tilde{O}(n^{1 - 2/p} \cdot \comdeg u)$ edges from $\ep$.
            Additionally,
            $\mu \in \Omega\left( \frac{|\ep|}{n} \right)$.
        \item For each $u \in V$ that is the tail of at least one
            $e \in \ebar \cup \ep$, exactly one $v \in \vlist$
            holds $\sentdeg u$.
    \end{itemize}
\end{definition}

\begin{definition}[$K_p$-Input Cluster]%
\label{def:centralized-input-cluster}
    A \emph{$K_p$-input cluster} $C$ is a $K_p$-compatible
    cluster in which $\ebar$ and $\ep$ are  distributed to $\vlist$ in
    the following manner. There is a $(\beta, V)$-vertex chain
    $\vchain{E} = \vlist$ with $\beta = n$, whose elements
    $\chainmem{E}{1}{},\dots,\chainmem{E}{k}{}$ are numbered in
    increasing order. Each $v = \chainmem{E}{i}{}$, for each
    $u \in \chainres{E}{v}{}$ knows all edges
    $e \in \ep \cup \ebar$ with their tail at $u$, and there are
    $\tilde{O}(n^{1 - 2/p} \cdot \comdeg v)$ such edges.
\end{definition}

We now state the main result of this section, that we can construct a
$(p,p')$-split $K_p$-partition tree.

\begin{theorem}%
\label{thm:congest-tree-algo-clique}
    Given a $K_p$-compatible cluster $C = (\vcluster, E_C)$ and
    constants $p, p'$, $p' \le p$, $p > 3$, there is a deterministic
    \congest algorithm on $C$ that, in $n^{1 - 2/p \rc}$ rounds,
    constructs a $(p',p)$-split $K_p$-partition tree $T$ for
    $V_1 = \vlist$, $V_2 = V \setminus \vlist$,
    $E_1 = E(\vlist, \vlist)$, $E_2 = \ep$, and $E_{12} = \ebar$ for
    $a,b = \tilde{O}(k^{1/p})$ so that all parts in $T$ are known to
    all $\vlist$.
\end{theorem}

We split the proof into two parts. In
Theorem~\ref{thm:central-to-tree}, we show that a $K_p$-input cluster
can construct the desired partition using
Theorem~\ref{thm:simulate-special}. In
Theorem~\ref{thm:input-to-central}, we show that a $K_p$-compatible
cluster can reorganize itself into an $K_p$-input cluster.

First, we provide the following broadcasting routine, used
in~\cite{CCGL20}.

\begin{lemma}%
\label{lem:n-message-broadcast}
    Given a $K_p$-compatible cluster
    $\genericCluster = (\genericClusterNodes, \genericClusterEdges)$
    and a set $M$ of $O(n)$ messages such that each message $m_i$ is
    initially known by exactly one $v \in \vlist$, there exists a
    deterministic \congest algorithm that finishes in
    $n^{1/2 \rc}$ rounds and makes all messages $m_i \in M$
    known to all $v \in \vlist$.
\end{lemma}

\begin{proof}
    Let $\hat v$ be the lowest-numbered vertex in $\vlist$. Each
    $v \in \vlist$ sends each message $m_i$ that it initially holds to
    $\hat v$. Since each $v \in \vlist$ sends and receives $O(n)$
    messages, this can be done in $n^{1/2 \rc}$ rounds.

    Now $\hat v$ knows the sequence $M$ of all $O(n)$ messages. We can
    make $M$ known to all $v \in \vlist$ in $n^{1/2 \rc}$
    rounds by performing $O(\log k)$ steps in which each
    $v \in \vlist$ that knows $M$ sends $M$ to a unique vertex
    $u \in \vlist$ that does not know $M$. In each step, a vertex
    sends and receives $O(n)$ messages, which can be done in
    $n^{1/2 \rc}$ rounds. Also, after each round, the
    number of $v \in \vlist$ that know $M$ doubles, so that all
    $v \in \vlist$ know $M$ after $O(\log n)$ rounds.
\end{proof}

\begin{theorem}%
\label{thm:central-to-tree}
    Given a $K_p$-input cluster $C = (\vcluster, E_C)$ and constants
    $p, p'$, $p' \le p$, $p > 3$, there is a deterministic \congest
    algorithm on $C$ that, in $n^{1/2 \rc}$ rounds, constructs a
    $(p',p)$-split $K_p$-partition tree $T$ for $V_1 = \vlist$,
    $V_2 = V \setminus \vlist$, $E_1 = E(\vlist, \vlist)$,
    $E_2 = \ep$, and $E_{12} = \ebar$ for $a,b = \tilde{O}(k^{1/p})$
    so that all parts in $T$ are known to all of $\vlist$.
\end{theorem}

As for $p = 3$, we show a \specialalgo constructing a layer of a
$(p,p')$-split $K_p$ partition tree. The algorithm makes use of
\auxtok{s}, as for $p > 3$ we can have $v \in \vlist$ that are
responsible for edges in $\ebar \cup \ep$ that emanate from a large
number of $u \in V \setminus \vlist$. To obtain a good runtime with
Theorem~\ref{thm:simulate-special}, we need a low $T_{\max}$, i.e.,
vertices hold few \maintok{s}. If we create a \maintok for each
$u \in V$, we might have some $v \in \vlist$ that are responsible for
a large number of $u \in V$ holding too many \maintok{s}. Therefore,
we instead aggregate the information about all $u \in V$ for which
$v \in \vlist$ is responsible into a single \maintok and encode the
information about individual $u$ in \auxtok{s}.

\begin{lemma}%
\label{lem:kp-htree-streaming}
    There exist constants $c_1,c_2$ such that, given a split graph
    $G$, some $2 \le p' < p$, $0 \le i < p$ levels of a $(p,p')$-split
    $K_p$-partition tree $T_{G',K_p}$ for $c_1,c_2$, and a part
    $U_{S,j}$ in the $i$-th level of $T$, $i = |S|$, there exists a
    \specialalgo with parameters $L = O(\polylog n)$,
    $\nin = \tilde{O}(n)$, $\nout = O(n^{1/p})$,
    $\paramq = O(n^{1/p})$, $\paramy = \nout$ that reads a stream $S$
    and produces a stream $R$, where:
    \begin{itemize}
        \item The elements of $R$ are the endpoints of intervals
            of vertex numbers that define a valid child partition of
            $U_{S,j}$ in some $(p,p')$-split $K_p$-partition tree for
            $c_1,c_2$.
        \item Let $W = V_2$ if $i + 1 < \pi$ and $W = V_1$ otherwise.
            Let $\sigma$ be an arbitrary partition of $W$ into
            $k = |V_1|$ contiguously-numbered intervals
            $\sigma_1,\dots,\sigma_k$. The input stream $S$ contains
            $k$ tokens $\tau_l$, $1 \le l \le k$, where the token
            $\tau_l$ encodes
            $\sum_{v \in \sigma_l}|E(v, V_1)|$,
            $\sum_{v \in \sigma_l}|E(v, V_2)|$, and
            $\sum_{v \in \sigma_l}|E(v, U')|$ for each
            $U' \in {\tt anc}(U_{S,j})$. The tokens
            $\auxt_{l,1},\dots,\auxt_{l,|\sigma_l|}$ contain
            $|E(v, U')|$, $|E(v, V_1)|$, and $|E(v, V_2)|$ for each
            $v \in \sigma_l$ in order of increasing vertex number.
    \end{itemize}
\end{lemma}

\begin{proof}
    We employ a counter-based algorithm (see
    Algorithm~\ref{alg:kp-streaming}), using counters for the three
    constraints \consdegbb, \consupdegbb, and \consdegba if
    $i + 1 < \pi$ and \consdegaa, \consupdegaa, \consupdegab
    otherwise.

    \begin{algorithm}[htbp]
        \caption{
            Construct one layer of a $(p,p')$-split
            $K_p$-partition tree%
        }%
        \label{alg:kp-streaming}
        \begin{algorithmic}[1]
            \State $l \gets 0$
            \State Set each counter to 0
            \For {$\tau_l \in S$}
                \State \readop $\tau_l$
                \State Extract sums from $\tau_l$
                \State Add each sum to the appropriate counter
                \If {any counter exceeds its maximum value}
                    \State Restore previous states of counters by
                        subtracting the sums from $\tau_l$
                    \State \auxop $\auxt_l$
                    \For {$\auxt_{l,m} \in \auxt_l$}
                        \State \readop $\auxt_{l,m}$
                        \State $r \gets$ the vertex number of
                            $\auxt_{l,m}$
                        \State Extract sums from $\auxt_{l,m}$
                        \State Add each sum to the appropriate counter
                        \If {any counter exceeds its maximum value}
                            \State \writeop $[l, r - 1]$
                            \State Reset all counters to 0
                            \State Add each sum to the appropriate
                                counter
                            \State $l \gets r$
                        \EndIf
                    \EndFor
                \EndIf
            \EndFor
            \State \writeop $[r, \max W]$
        \end{algorithmic}
    \end{algorithm}

    We show that this \specialalgo satisfies
    Lemma~\ref{lem:kp-htree-streaming}. Each token contains
    $O(\log n)$ numbers, each polynomial in $n$, so
    $L = O(\polylog n)$. Because the partition $\sigma$ has $O(n)$
    parts, and there is one \maintok for each, $\nin = O(n)$. Each
    output token represents a part of a partition in a $(p,p')$-split
    $K_p$-partition tree, so by Definition~\ref{def:split-tree},
    $\nout = O(n^{1/p})$. The algorithm only performs \auxop when a
    new part is formed, which is also bounded by the number of parts
    $\paramq \le \nout = O(n^{1/p})$. Finally, because the total
    number of \writeop operations performed by the algorithm is
    $\nout$, the number of \writeop operations between any two \readop
    operations is also bounded by $\paramy = \nout$.

    This constructs a valid child partition of $U_{S,j}$. A vertex can
    always join a new part on its own without causing any counter to
    exceed its maximum value, due to the `$+k$' terms in the bounds of
    \consdegaa, \consupdegaa and the `$+n$' terms in those of
    \consdegbb, \consupdegbb, \consdegba, \consupdegab. Starting a new
    part each overflow guarantees properties 1--6 of
    Definition~\ref{def:split-tree}. It remains is to bound the number
    of times a counter can exceed its maximum value:
    \begin{enumerate}
        \item A part closed due to counter \consdegbb has
            at least $c_1m_2/b$ edges into $V_2$, of which there are
            at most $2m_2$ total (counting each endpoint twice). There
            can be at most $2b/c_1$ such parts.
        \item A part closed due to counter \consupdegbb
            has at least $c_2i\tilde{m}_2/b^2$ edges to its ancestor
            parts; due to constraint 1, the total count of such edges
            is at most $i(c_1m_2/b + n) \leq i(c_1 + 1)\tilde{m}_2/b$,
            implying at most $(c_1 + 1)b/c_2$ such parts.
        \item A part closed due to counter \consdegba has
            at least $c_1m_{12}/b$ edges into $V_1$, of which there
            are at most $2m_{12}$ total (counting each endpoint
            twice). There can be at most $2b/c_1$ such parts.
        \item A part closed due to counter \consdegaa has
            at least $c_1m_1/a$ edges into $V_1$, of which there are
            at most $2m_1$ total (counting each endpoint twice). There
            can be at most $2a/c_1$ such parts.
        \item A part closed due to counter \consupdegaa
            has at least $c_2(i - \pi)\tilde{m}_1/a^2$ edges into its
            ancestor parts at depth at least $\pi$, of which there are
            at most
            $
                (i - \pi)(c_1m_1/a + k) \leq
                (i - \pi)(c_1 + 1)\tilde{m}_1/a
            $.
            There can therefore be at most $(c_1 + 1)a/c_2$ such
            parts.
        \item A part closed due to counter \consupdegab
            has at least $c_2\pi\tilde{m}_{12}/(ab)$ edges to its
            ancestor parts at depth less than $\pi$, of which there
            are at most
            $p'(c_1m_{12}/b + n) \leq p'(c_1 + 1)\tilde{m}_{12}/b$.
            Thus, there are at most $(c_1 + 1)a/c_2$ such
            parts.
    \end{enumerate}
    Therefore, if there exist constants $c_1,c_2$ such that
    $2a/c_1 + (c_1 + 1)a/c_2 + (c_1 + 1)a/c_2 < a$ and
    $2b/c_1 + (c_1 + 1)b/c_2 + 2b/c_1 < b$,
    the counter-based process is guaranteed to produce a valid
    partition. Indeed, $c_1 = 8, c_2 = 36$ satisfy this property.
\end{proof}

Theorem~\ref{thm:simulate-special} can
simulate the algorithm described in Lemma~\ref{lem:kp-htree-streaming}
in $n^{1/2 \rc}$ rounds, as such.

\begin{lemma}%
\label{lem:kp-htree-one-layer}
    There exist constants $c_1,c_2$ such that, given a $K_p$-input
    cluster
    $\genericCluster = (\genericClusterNodes, \genericClusterEdges)$
    in graph $G = (V, E)$ for $\phi \in O(\polylog n)$ and
    $\delta = n^{1/p}$, where each $v \in \vlist$ knows the first
    $0 \le i < p$ layers of a $(p,p')$-split $K_p$-partition tree
    $T_{G,K_p}$ for $V_1 = \vlist$, $V_2 = V \setminus \vlist$,
    $E_1 = E(\vlist, \vlist)$, $E_2 = \ep$, $E_{12} = \ebar$,
    $a,b = \tilde{O}(k^{1/p})$, and some $2 \le p' \le p$ with
    constants $c_1,c_2$, there exists a deterministic \congest
    algorithm that finishes in $n^{1 - 2/p \rc}$ rounds and
    constructs a next layer for the tree $T_{G,K_p}$ so that all parts
    in the next layer are known to all $v \in \vlist$.
\end{lemma}

\begin{proof}
    We invoke \Cref{thm:simulate-special} on the streaming algorithm
    implied by Lemma~\ref{lem:kp-htree-streaming}. We execute several
    instances of this \specialalgo in parallel, one for each part of
    the lowest completed level of the $(p,p')$-split $K_p$-partition
    tree $T_{G,K_p}$. If $i + 1 \ge \pi$, $W = V_1 = \vlist$ and we
    can let each part $\sigma_i$ of the partition $\sigma$ correspond
    to $\chainmem{E}{i}{}$, so that each token $\tau_i$ is associated
    with a single auxiliary token $\alpha_{i,1}$ that holds the given
    information about one vertex. Since each vertex $v \in \vlist$
    knows all prior layers of the tree $T_{G,K_p}$ and its own edges,
    it holds the $O(1)$ tokens that describe it for each of the
    $O(n^{1 - 1/p})$ parts in layer $0 \le i < p$. If instead
    $i + 1 < \pi$, $W = V_2 = V \setminus V_1$ and we let $\sigma_i$
    be the interval of vertices corresponding to
    $\chainres{E}{\chainmem{E}{i}{}}{}$. Since each $v \in \vlist$
    knows all the past layers of the $(p,p')$-split $K_p$-partition
    tree and is only responsible for producing tokens that describe
    the vertices for which it knows
    $E(u, V) \cap (\ebar \cup \ep)$, we again find that it has
    all the necessary information to produce these tokens locally.

    Applying \Cref{thm:simulate-special} with
    parameters $T_{\max} \in O(1), \lambda = 1$ gives that all
    $O(n^{1 - 1/p})$ algorithms finish in parallel in
    $
        (
            \frac{1}{n^{1 - 2/p}}(k^{1 - 1/p} + k)
            + (n^{1/p} + 1)(1 + \frac{n^{1 - 1/p}}{n^{1 - 2/p}})
        ) \rcf
        =
        n^{2/p \rc}
        =
        n^{1/2 \rc}
    $
    rounds as desired. Together, the output streams of all of the
    executed algorithms contain all $O(n)$ parts in the next layer of
    $T$. We can distribute these parts so that they are known to all
    $v \in \vlist$ by applying Lemma~\ref{lem:n-message-broadcast}.
\end{proof}

We are now ready to prove Theorem~\ref{thm:central-to-tree}.

\begin{proof}[Proof of Theorem~\ref{thm:central-to-tree}]
    We apply Lemma~\ref{lem:kp-htree-one-layer} $p$ times, for a total
    round complexity of
    $n^{1/2 \rc}p = n^{1/2 \rc}$. Note that,
    because Lemma~\ref{lem:kp-htree-one-layer} guarantees that all
    parts of the new layer of $T$ are known to all $v \in \vlist$, no
    additional work is required to satisfy the preconditions of
    Lemma~\ref{lem:kp-htree-one-layer} between layers.
\end{proof}

Showing the following theorem concludes the proof of
Theorem~\ref{thm:congest-tree-algo-clique}.

\begin{theorem}%
\label{thm:input-to-central}
    Given a $K_p$-compatible cluster $C = (\vcluster, E_C)$, there
    exists a deterministic \congest algorithm such that, after
    $n^{1 - 2/p \rc}$ rounds, the input $\ep$ is
    held in a manner satisfying
    Definition~\ref{def:centralized-input-cluster}.
\end{theorem}

\begin{proof}
    Using Lemma~\ref{lem:n-message-broadcast}, make $\sentdeg u$ for
    all $u \in V$ known to all $\vlist$ within $n^{1/2 \rc}$ rounds.
    Now, all $\vlist$ locally find a vertex chain $\vchain{E}$
    satisfying Definition~\ref{def:centralized-input-cluster}. For
    each $(u, w) \in \ebar \cup \ep$ held by a $v \in \vlist$, $v$
    sends $(u, w)$ to $\chainass{E}{u}{}$. As each vertex sends and
    receives $n^{1 - 2/p}$ messages, this takes $n^{1 - 2/p \rc}$
    rounds.
\end{proof}


\section{Deterministic Triangle Listing in \texorpdfstring{$
    \boldsymbol{n^{1/3 \rc}}
$} {n\textasciicircum1/3 + o(1)} Rounds}%
\label{sec:triangles}

We are now ready to use $K_3$-partition trees to optimally list
triangles in \congest. 

\begin{theorem}[Triangle Listing in \congest]%
\label{thm:congest-tri-find}
    Given a graph $G = (V, E)$ with $n = |V|$ vertices, there exists a
    deterministic \congest algorithm that completes in
    $n^{1/3 \rc}$ rounds and lists all triangles in $G$.
\end{theorem}

We split the proof of Theorem~\ref{thm:congest-tri-find}.
Lemma~\ref{lem:k3-wrapper} reduces the problem from general graphs to
listing triangles in a suitable communication cluster. This reduction
is described in~\cite{CS20}, and we provide a very brief summary here.
Lemma~\ref{lem:k3-in-cluster} shows that triangle listing can be
performed optimally in suitable communication clusters, which we prove
by applying the $K_3$-partition-tree construction algorithm of
Theorem~\ref{thm:congest-tree-algo-tri} to achieve optimal load
balancing.

\begin{lemma}%
\label{lem:k3-wrapper}
    Given a graph $G = (V, E)$, suppose for some function $f$ there
    exists a deterministic \congest algorithm that allows any
    $K_3$-compatible cluster on $G$ to list in $\tilde{O}(f(n))$
    rounds all instances $H'$ of $K_3$ such that all three vertices in
    $H'$ lie in $C$. Then there is a deterministic \congest algorithm
    on $G$ that runs in $\tilde{O}(f(n))$
    rounds and
    outputs all triangles in $G$.
\end{lemma}

The high-level idea of the reduction as presented in~\cite{CS20} is to
run the deterministic expander decomposition of
Theorem~\ref{thm:cs-decomposition} on $G$ to obtain a number of
clusters $C_i$, list all triangles in each cluster in
$\tilde{O}(f(n))$ rounds, and then recurse on $E\rem$. Since $E\rem$
is bounded by a constant fraction of $n$, the recursion depth is
logarithmic in $n$; therefore, the overall runtime is
$\tilde{O}((f(n)\log n) = \tilde{O}(f(n))$.

More specifically, after obtaining sets $E_i^-$ and $\vinnerof{C_i}$
from the partition
$E = E_1 \cup E_2 \cup \cdots \cup E_x \cup E\rem$ produced by
Theorem~\ref{thm:cs-decomposition} as described in
Section~\ref{sec:preliminaries},~\cite{CS20} construct a set of
edges $E_i^+$ for each $E_i$. They then prove that
\begin{itemize}
    \item All $G[E_i^+]$ have high conductance and, in our notation,
        are therefore $K_3$-compatible clusters.
    \item For any edge $e \in E(\vinnerof{C_i}, \vinnerof{C_i})$, any
        triangle in
        $G[\bigcup_{1 \le i \le x} E_i]$ that contains $e$ is
        contained in $G[E_i^+]$.
    \item Each edge $e \in E$ is in at most two subgraphs $G[E_i^+]$,
        so that all subgraphs can use Theorem~\ref{thm:cs-routing} to
        route messages internally with only a constant overhead
        factor.
\end{itemize}

We can list all triangles in $G[E_i^+]$ by taking the $K_3$-compatible
cluster $C_i = G[E_i^+]$, letting $\vlistof{C_i}$ be the set of all
$v \in V_{C_i}$ with $\deg_{C_i}(v) \ge |V_{C_i}|^{1/3}$. Applying the
presupposed algorithm for listing on clusters outputs all triangles in
$G[\bigcup_{1 \le i \le x} E_i]$ in $\tilde{O}(f(n))$ rounds. By
Lemma~\ref{lem:remaining-edges-small},
recursing on $G[E \setminus \bigcup_{1 \le i \le x} E_i]$ produces a
logarithmic recursion depth, as desired.

\begin{lemma}%
\label{lem:k3-in-cluster}
    Given a graph $G = (V, E)$ and a $K_3$-compatible cluster
    $\genericCluster = (\genericClusterNodes, \genericClusterEdges)$,
    there is a deterministic \congest algorithm on $C$ that finishes
    in $k^{1/3} \rcf$ rounds and lists all triangles in
    $C$.
\end{lemma}

\begin{proof}

    We divide the proof into two cases: triangles containing a vertex
    in $\vld = V_C \setminus \vlist$, and triangles with all three
    vertices in $\vlist$.

    \paragraph{\textbf{Low-degree vertices.}} We can list all $K_3$
    that contain low-degree vertices using an exhaustive search
    procedure. Lemma~\ref{thm:exhaustive-search} gives the guarantees
    of this process.

    \begin{lemma}[{\hspace{-1pt}\cite[Claim~19]{CFLLO}}]%
    \label{thm:exhaustive-search}
        Given a graph $G = (V, E)$, and some value $\alpha$, every
        vertex
        $v$ with $\deg(v) \le \alpha$ can deterministically learn its
        induced 2-hop neighborhood in $O(\alpha)$ rounds of
        \congest.
    \end{lemma}

    We apply Lemma~\ref{thm:exhaustive-search} to $v \in \vld$
    with $\alpha = K^{1/3}$ to list all triangles that contain a
    vertex in $\vld$ in $O(K^{1/3})$ rounds.
    
    \paragraph{\textbf{High-degree vertices.}} We were tasked with
    listing all triangles in $C$. We have now listed all triangles in
    $C$ involving at least one vertex in $\vld$, and are thus left
    with listing all triangles in $C$ for which all their vertices
    have degree at least $K^{1/3}$. To do so, we construct a
    $K_3$-partition tree over $C[\vlist]$.

    As $C$ is a $K_3$-compatible cluster, we can use
    \Cref{thm:congest-tree-algo-tri} in order to construct a
    $K_3$-partition tree $T$ over $C[\vlist]$. By
    \Cref{thm:congest-tree-algo-tri}, $T$ is constructed such that its
    root and middle layers are known to all of $\vlist$, and its leaf
    layer is dispersed across
    $\vhd = \{v \in \vlist : \comdeg v \ge \frac{1}{2}\mu\}$ such that
    every $v \in \vhd$ knows $O(\frac{1}{\mu}\comdeg v)$ parts of this
    layer.

    By \Cref{thm:cfg-tree-can-list}, if each $v \in \vhd$ learns all
    the edges that cross some parts in ${\tt anc}(U_{S,i})$ for
    each part $U_{S,i}$ known by $v$, then the vertices $\vhd$ will be
    able to list all triangles in $C[\vlist]$. By
    Definition~\ref{def:htree-congest}, the number of such edges is
    $O(k^{1/3}\mu + k^{2/3})$, which, since
    $\mu \ge K^{1/3} \ge k^{1/3}$, the number of such edges is
    $O(k^{1/3}\mu)$. Each vertex $v \in \vhd$ is assigned
    $O(\frac{1}{\mu}\comdeg v)$ parts, so it must learn
    $O(k^{1/3}\comdeg v)$ edges in total.\footnote{%
        Since we only use vertices $v$ with
        $\comdeg v \ge \frac{1}{2}\mu$, the rounding error from the
        divisions incurs at most a factor of two and does not affect
        the asymptotics of the number of edges each vertex must learn.
    }
    Each vertex $v \in \vhd$ uses the following process to learn the
    required edges for each leaf part $U_{S,i}$ it is assigned in
    parallel.
    \begin{enumerate}
        \item For each $U \in {\tt anc}(U_{S,i})$, for each
            $u \in U$, $v$ sends the vertex numbers that define the
            endpoints of the intervals
            ${\tt anc}(U_{S,i}) \setminus U$ to $u$.
        \item Each $u \in \vlist$ that receives an interval $U$ from a
            $v \in \vhd$ replies by sending all of the edges $E(u, U)$
            to $v$.
    \end{enumerate}

    During Step~1, each vertex $v \in \vhd$ sends $O(1)$ messages to
    each of the vertices in the $O(1)$ partitions
    $U \in {\tt anc}(U_{S,i})$ for each of the
    $O(\frac{1}{\mu}\comdeg v)$ leaf parts $U_{S,i}$ it is assigned.
    By Definition~\ref{def:htree-congest}, each part $U$ has
    $O(k^{2/3})$ vertices, so that the total number of messages sent
    by each $v \in \vhd$ is
    $O(\frac{1}{\mu}k^{2/3}\comdeg v) = O(k^{1/3}\comdeg v)$. Also,
    during Step~1, each vertex $u \in \vlist$ receives $O(1)$ messages
    for each leaf part of which it is a member, $O(k^{1/3})$ messages
    for each middle layer part of which it is a member, and
    $O(k^{2/3})$ messages for each root partition part of which it is
    a member. Since it is a member of $O(k^{2/3})$, $O(k^{1/3})$, and
    $O(1)$ of each type of part, respectively, the total number of
    messages received by each $v \in \vlist$ is $O(k^{2/3})$.
    Therefore, Step~1 completes in $k^{1/3} \rcf$ rounds.

    During Step~2, we have already established that each $v \in \vhd$
    receives $O(k^{1/3}\comdeg v)$ responses. We now show that each
    vertex $u \in \vlist$ sends each of its edges to some $v \in \vhd$
    at most $O(k^{1/3})$ times. Consider an edge $(u, w)$. We proceed
    by casework on the level of the part $U$ for which $u \in U$ and
    count the number of parts $W \ne U$ for which $w \in W$ and
    $U \in {\tt anc}(W)$ or $W \in {\tt anc}(U)$.
    \begin{enumerate}
        \item There is only one part $U$ in the root partition
            $U_\emptyset$ that contains $u$; call this part $U_0$.
            Vertex $w$ is in exactly one child part $W$ of $U_0$, and
            $W$ has $O(k^{1/3})$ children. For each of the
            $O(k^{1/3})$ child parts of $U_0$, $w$ is in exactly one
            of their children. $U_0$ contributes $O(k^{1/3})$
            occasions to send the edge $(u, w)$.
        \item There are $O(k^{1/3})$ first-layer parts $U$ that
            contain $u$. Each such part has exactly one child that
            contains $w$, and at most one is descended from the unique
            root-partition part $W_0$ that contains $w$. The first
            layer contributes $O(k^{1/3})$ occasions to send the edge
            $(u, w)$.
        \item There are $O(k^{2/3})$ leaf-layer parts that contain
            $u$. Each of the $O(k^{1/3})$ first-layer parts that
            contain $w$ are an ancestor of one such part, and
            $W_0$ is an ancestor of $O(k^{1/3})$ such parts. The leaf
            layer contributes $O(k^{1/3})$ occasions to send the edge
            $(u, w)$.
    \end{enumerate}
    Therefore, Step~2 also completes in $k^{1/3} \rcf$
    rounds.
\end{proof}

We are now ready to prove Theorem~\ref{thm:congest-tri-find}.

\begin{proof}[Proof of Theorem~\ref{thm:congest-tri-find}]
    By Lemma~\ref{lem:k3-in-cluster}, there exists a \congest
    algorithm that lists all triangles in a $K_3$-compatible cluster
    $C$ in $k^{1/3}\rcf = n^{1/3 \rc}$
    rounds. Therefore, by Lemma~\ref{lem:k3-wrapper}, there is a
    deterministic \congest algorithm that lists all triangles in $G$
    in $n^{1/3 \rc}$
    rounds, as desired.
\end{proof}

\section{Deterministic \texorpdfstring{$
    \boldsymbol{K_p}
$}{K\_p} Listing in \texorpdfstring{$
    \boldsymbol{n^{1 - 2/p \rc}}
$}{n\textasciicircum(1--2/p+o(1))} Rounds}%
\label{sec:cliques}

We now set out to use $(p,p')$-split $K_p$-partition trees to
optimally list cliques of size greater than 3. 

\begin{theorem}[$K_p$ Listing in \congest, $p > 3$]%
\label{thm:congest-clique-find}
    Given a graph $G = (V, E)$ with $n = |V|$ vertices and a constant
    $p > 3$, there is a deterministic \congest algorithm that, in
    $n^{1 - 2/p \rc}$ rounds, lists all $K_p$ in $G$.
\end{theorem}

As in the case of triangles, we decompose the proof into two parts
following the high-level organization of~\cite{CCGL20}: a reduction
from general graphs to $K_p$-compatible clusters, and listing in
clusters. However, the guarantees of the deterministic expander
decomposition of \hspace{-1pt}\cite{CS20} differ than those of the
randomized one, which is why we cannot take the previous outline as
is. Our most significant contribution in this part is using our
deterministic partition trees construction to enable efficient listing
within clusters. We begin by proving the following lemma, which shows
this.
\begin{lemma}%
\label{lem:congest-expander-find-clique}
    Given a graph $G = (V, E)$ and a $K_p$-compatible cluster
    $\genericCluster = (\genericClusterNodes, \genericClusterEdges)$,
    there is a deterministic \congest algorithm on $C$ that finishes
    in $n^{1 - 2/p \rc}$ rounds and lists all $p$-cliques
    $H' = \{v_1, \dots, v_p\}$ in $G$ such that there exists some
    $2 \le p' \le p$ for which the vertices of $H'$ can be partitioned
    into $V_{p'} = \{v_1, \dots, v_{p'}\} \subseteq \vlist$ and
    $
        V_{p \setminus p'}
        = \{v_{p' + 1}, \dots, v_p\} \subseteq V \setminus \vlist
    $
    so that $V_{p'} \times V_{p'} \subseteq E(\vlist, \vlist)$,
    $V_{p'} \times V_{p \setminus p'} \subseteq \ebar$, and
    $V_{p \setminus p'} \times V_{p \setminus p'}\subseteq \ep$.
\end{lemma}

\begin{proof}
    Fix some $2 \le p' \le p$ and list all $K_p$ in
    $G[E(\vlist, \vlist) \cup \ebar \cup \ep]$ with exactly $p'$
    vertices in $\vlist$ in $n^{1 - 2/p \rc}$ rounds. As
    $p = O(\log n)$, repeating this process for all $2 \le p' \le p$
    gives the desired algorithm.

    We begin by invoking Theorem~\ref{thm:congest-tree-algo-clique} to
    construct a $(p,p')$-split $K_p$-partition tree $T_{p'}$ for
    values $a = b = \lceil k^{1/p} \rceil$. Since $k \ge 1$, we known
    that $\lceil k^{1/p} \rceil \le 2 k^{1/p}$, so that
    $a,b = O(k^{1/p})$ and the $(p,p')$-split $K_p$ partition tree can
    be constructed in $\tilde{O}(n^{1 - 2/p \rc})$ rounds. After this,
    all vertices $v \in \vlist$ know all parts of partition tree
    $T_{p'}$. Since $a, b \le 2 k^{1/p}$, the number of leaf parts is
    at most $a^{p'}b^{p - p'} \le 2^pk$. For constant $p$, we
    therefore have $O(k)$ leaf parts. Let each $v \in \vlist$ forget
    all but $O(1)$ parts of the leaf layer so that each part is known
    by exactly on $v \in \vlist$; the vertices can locally compute
    which parts to remember in a predetermined fashion. Now, we apply
    Lemma~\ref{lem:k3-chain-distribute} to redistribute the parts of
    the leaf layer so that each $v \in \vhd$ knows
    $O(\frac{1}{\mu}\comdeg v)$ parts in $k^{1/3 \rc}$ rounds.

    The rest of the algorithm proceeds as in~\cite{CCGL20}, using the
    assigned $(p,p')$-split $K_p$-partition tree leaves in the place
    of a randomized partition. Before performing the final listing
    step, we again redistribute the edges $\ep$ so that each vertex
    $v \in \vgeqavg$ holds
    $O(\frac{|\ep|}{k\mu}\comdeg v) = O(n^{1 - 2/p}\comdeg c)$
    edges. Each vertex $v \in \vlist$ can send in
    $n^{1/2 \rc}$ rounds to all other vertices in $\vlist$
    how many edges from $\ep$ it currently holds. This information is
    sufficient for all vertices to deterministically compute which
    vertices they should send edges to, and, since each $v \in \vlist$
    initially holds $\tilde{O}(n^{1 - 2/p} \cdot \deg_C(v))$ edges
    from $\ep$, the redistribution can finish in
    $n^{1 - 2/p \rc}$ rounds.

    Finally, each vertex $v \in \vlist$ sends each of the edges $e$ it
    knows to each vertex $u \in \vhd$ that is assigned a leaf part
    $U_{S,i}$ for which $e$ crosses two ancestors of $U_{S,i}$.
    Constraints \consupdegbb, \consupdegaa, and \consupdegab of
    Definition~\ref{def:split-tree} guarantee that each $u$ receives
    $\tilde{O}(E(\vlist, \vlist)/a^2 + k)$ edges from
    $E(\vlist, \vlist)$, $\tilde{O}(|\ebar|/(ab) + n)$ edges
    from $\ebar$, and $\tilde{O}(|\ep|/b^2 + n)$ edges from $\ep$.
    Since $|\ebar| = O(|E(\vlist, \vlist)|)$ and
    $|\ep| = O(n\mu)$, we have that the total number of edges learned
    by $u$ is $\tilde{O}(n^{1 - 2/p}\mu + n)$. Since
    $\mu \ge n^{1/2}$, this number is $\tilde{O}(n^{1 - 2/p}\mu)$, and
    therefore each $u \in \vhd$ can receive the edges for all of the
    $O(\frac{1}{\mu}\comdeg u)$ parts it holds in
    $n^{1 - 2/p \rc}$ rounds. Since the total number of
    edges sent over the graph is $\tilde{O}(k\mu n^{1 - 2/p})$ and
    each edge is sent an equal number of times, it holds that each
    edge in $E(\vlist, \vlist)$, $\ebar$, $\ep$ is sent
    $\tilde{O}(n^{1 - 2/p})$, $\tilde{O}(n^{1 - 2/p})$, and
    $\tilde{O}(\frac{k\mu}{|\ep|}n^{1 - 2/p})$ times, respectively.
    Since each $v \in \vlist$ holds $O(\comdeg v)$ edges in
    $E(\vlist, \vlist)$ and $\ebar$ and
    $\tilde{O}(\frac{|\ep|}{k\mu}n^{1 - 2/p})$ edges in $\ep$, each
    $v \in \vlist$ can send all of the necessary edges in
    $n^{1 - 2/p \rc}$ rounds.
\end{proof}

We now state the two main lemmas that permit us to use
Lemma~\ref{lem:congest-expander-find-clique} to prove
Theorem~\ref{thm:congest-clique-find}. They separate the task of
reducing general clique listing to clique listing on clusters for
$p > 4$ and $p = 4$, respectively.

\begin{lemma}%
\label{lem:kp-wrapper}
    Given a graph $G = (V, E)$, if there exists a deterministic
    \congest algorithm that allows any $K_p$-compatible cluster on $G$
    for $p > 4$ to list in $n^{1 - 2/p \rc}$ rounds all
    instances $H'$ of $K_p$ in
    $G[E(\vlist, \vlist) \cup \ebar \cup \ep]$ such that at least
    one edge of $H'$ lies in $C$, then there is a deterministic
    \congest algorithm on $G$ that runs in
    $n^{1 - 2/p \rc}$ rounds and outputs all instances of
    $K_p$ in $G$.
\end{lemma}

\begin{lemma}%
\label{lem:k4-wrapper}
    Given a graph $G = (V, E)$, if there exists a deterministic
    \congest algorithm that allows any $K_4$-compatible cluster on $G$
    to list in $n^{1/2 \rc}$ rounds all instances $H'$ of
    $K_4$ in $G[E(\vlist, \vlist) \cup \ebar \cup \ep]$ such that
    at least one edge in $H'$ lies in $C$, then there is a
    deterministic \congest algorithm on $G$ that runs in
    $n^{1/2 \rc}$ rounds and outputs all instances of $K_4$
    in $G$.
\end{lemma}

We prove the lemmas in
\Cref{sec:wrapper-clique-5,sec:wrapper-clique-4}, respectively. Using
these lemmas, we can prove Theorem~\ref{thm:congest-clique-find}.

\begin{proof}[Proof of Theorem~\ref{thm:congest-clique-find}]
    By Lemma~\ref{lem:congest-expander-find-clique}, there exists a
    deterministic algorithm on $K_p$-compatible clusters that
    satisfies the preconditions of Lemmas~\ref{lem:kp-wrapper}
    and~\ref{lem:k4-wrapper}. Therefore, for $p > 4$ and for $p = 4$,
    there exists a deterministic \congest algorithm on $G$ that
    finishes in $n^{1 - 2/p \rc}$ rounds and lists all
    $K_p$ in $G$.
\end{proof}

\subsection{\texorpdfstring{$
    \boldsymbol{K_p}
$}{K\_p} Listing for \texorpdfstring{$
    \boldsymbol{p > 4}
$}{p \textgreater\ 4}}%
\label{sec:wrapper-clique-5}

We prove Lemma~\ref{lem:kp-wrapper} by employing the listing algorithm
of Lemma~\ref{lem:congest-expander-find-clique}. The overall structure
is to split $G$ into high-conductance clusters and prove that the
conditions of a $K_p$-compatible cluster are met for a majority of
clusters. Remaining clusters are deferred for later, and we recurse on
the graph they induce, which has at most half the edges of $G$.

We begin by invoking the deterministic expander decomposition stated
in Theorem~\ref{thm:cs-decomposition} on $G$, for some constant
$\epsilon$ that we will fix later and
$
    \phi = \operatorname{poly}(\epsilon)
        \cdot 2^{-O(\sqrt{\log n \log\log n})}
$.
The decomposition takes $n^{o(1)}$ rounds and produces a partition of
the edges $E = E_1 \cup E_2 \cup \cdots \cup E_x \cup E\rem$, where
$G\left[ E_1 \right],\dots, G\left[ E_x \right]$ are vertex-disjoint,
have conductance at least $\phi$, and
$\left|E\rem\right| \leq \epsilon|E|$. We construct
$\vinnerof{C_i}$ and $E_i^-$ as presented in
Section~\ref{sec:preliminaries}.
We also denote $E_i^+ = E_i \cup E(\vinnerof{C_i}, \vinnerof{C_i})$.
The following shows that
$G[E_i^+]$ have sufficiently high conductance to be communication
clusters.

\begin{lemma}%
    \label{lem:e+4-high-conductance}
    For all $1 \leq i \leq x$, $\Phi(G[E_i^+]) \geq \frac12\phi$.
\end{lemma}

\begin{proof}
    We recall that conductance of a graph $G = (E, V)$ is defined as
    the minimum value of
    \[
        \frac{|\partial S|}{\min\{\Vol(S), \Vol(V \setminus S)\}}
    \]
    over all $S \not\in \{\emptyset, V\}$, where
    $\Vol(S) = \sum_{v \in S}\deg(v)$ and
    $\partial S = E(S, V \setminus S)$. First, we observe that
    $G[E_i]$ and $G[E_i^+]$ have the same set of vertices $V_i$ and
    recall that $\Phi(G[E_i]) \geq \phi$. Since no edges are removed
    from $G[E_i]$ to form $G[E_i^+]$, we have that $|\partial S|$ in
    $G[E_i^+]$ is at least $|\partial S|$ in $G[E_i]$ for all cuts $S$
    of $V_i$. Also, $\deg(v)$ is greater in $G[E_i^+]$ than in
    $G[E_i]$ only if $v \in \vinnerof{C_i}$, from which we have that $\deg(v)$ in
    $G[E_i^+]$ is at most double $\deg(v)$ in $G[E_i]$. Therefore, the
    value of $\Vol(S)$ in $G[E_i^+]$ is at most double the value of
    $\Vol(S)$ in $G[E_i]$ for all cuts $S$ of $V_i$, from which it
    follows that $\Phi(G[E_i^+]) \geq \frac12\phi$.
\end{proof}

We now define for each $E_i^+$ a cluster
$C_i = (\vclusterof{C_i}, E_{C_i})$ for
which $E_{C_i} = E_i^+$
and $C_i$ is the subgraph induced by $E_{C_i}$, $C_i = G[E_{C_i}]$.
For some constant $\beta > 1$ that we fix later, we let
$\vlistof{C_i} \subseteq \vinnerof{C_i}$ be the set of vertices
$v \in \vinnerof{C_i}$ satisfying $\deg_C(v) > \beta n^{1 - 2/p}$.
Note that the sets of vertices $\vclusterof{C_i}$ and edges $E_{C_i}$
are pairwise disjoint over $C_i$.

We employ Lemma~\ref{lem:remaining-edges-small}, which bounds the
number of edges outside the union of all
$E(\vinnerof{C_i},\vinnerof{C_i})$ to
$3\epsilon|E|$, when we later choose the value of $\epsilon$. This
allows us to focus on listing cliques that have at least one edge in
$E(\vinnerof{C_i}, \vinnerof{C_i})$ for some $C$, since the total
number of
remaining edges is small enough for the recursion to have logarithmic
depth. In what follows, we focus on individual clusters and
for convenience drop the subscript $i$ from properties of the cluster
$C = C_i$. We demonstrate how to list most instances of
$K_p$ that have at least two vertices in some $\vinner$.

We list all cliques that have a vertex
$v \in \vinner \setminus \vlist$ in $O(n^{1 - 2/p})$ rounds using
exhaustive search.

\begin{lemma}%
    \label{lem:list-low-degree}
    Listing all instances of $K_p$ containing
    $v \in \vinner \setminus \vlist$ can be done in $O(n^{1 - 2/p})$
    rounds by an exhaustive search procedure. This can be done in
    parallel for all such vertices $v$.
\end{lemma}
\begin{proof}
    Since $\deg(v) \leq 2\deg_C(v)$ for all $v \in \vlist$, we have
    that the minimum degree of $v \in \vinner \setminus \vlist$ over
    all $C$ is at most $2\beta n^{1 - 2/p}$. We can therefore apply
    Lemma~\ref{thm:exhaustive-search} with
    $\alpha \in O(n^{1 - 2/p})$ to all instances of $K_p$ intersecting
    $\vinner \setminus \vlist$ in $O(n^{1 - 2/p})$ rounds.
\end{proof}

Lemma~\ref{lem:list-low-degree} leaves us only with the task of
listing cliques that have edges in $E(\vlist, \vlist)$. Note that in
any cluster $C$ with $|\vcluster| < \beta n^{1 - 2/p}$, all
$v \in \vinner$ satisfy $\comdeg v < \beta n^{1 - 2/p}$, from which
$\deg v < 2\beta n^{1 - 2/p}$. Therefore, $\vlist$ is empty for such
small clusters and they are eliminated by
Lemma~\ref{lem:list-low-degree}. We therefore are only left with
clusters $C$ satisfying $|\vcluster| \ge \beta n^{1 - 2/p}$.

Consider vertices $V \setminus \vlist$ (including, importantly,
$\vcluster \setminus \vlist$), and denote
\[
    S_C^* = \{
        u \not\in \vlist :
        1 \leq \comdeg u < \deg_{V \setminus \vlist}(u) / n^{1 - 2/p}
    \}
\]

Note that $u \in S_C^*$ have $\deg(u) > n^{1 - 2/p}$ and
$\deg_{\vlist}(u) < n^{2/p}$, so that
$|E(\vlist, S_C^*) < n^{1 + 2/p}|$. We call a vertex in $\vlist$
\emph{bad} if it has more than $n^{1 - 2/p}$ neighbors in $S_C^*$
and denote these vertices
$S_C = \{v \in \vlist : \deg_{S_C^*}(v) > n^{1 - 2/p}\}$. These
pose difficulties for learning information about edges outside the
cluster, so we show that the number of edges between them
is small and can be deferred to later recursive iterations.

\begin{lemma}%
    \label{lem:sc2-is-small}
    For constant $\beta > 1$ and $p \geq 5$,
    $
        \sum_{C : |\vcluster| > \beta n^{1 + 2/p}}
            |S_C|^2 \leq (4/\beta)|E|
    $
\end{lemma}
\begin{proof}
    The total number of edges $|E|$ is at least the number of edges
    that touch vertices in $S_C^*$. We find a lower bound for this by
    counting edges that are incident to vertices in $S_C^*$ but not
    $\vlist$. We sum $\deg_{V \setminus \vlist}(u)$ over all
    $u \in S_C^*$ and divide by 2, due to possible overcounting:
    \begin{align*}
        2|E|
        & \geq \sum_{u \in S_C^*} \deg_{V \setminus \vlist}(u)
          > n^{1 - 2/p}\sum_{u \in S_C^*} \deg_{C}(u) \\
        & \geq n^{1 - 2/p} \sum_{v \in S_C} \deg_{S_C^*}(v)
          > n^{1 - 2/p} \cdot n^{1 - 2/p} \cdot |S_C|
          = n^{2 - 4p}|S_C|
    \end{align*}
    from which $|S_C| < 2|E|/n^{2 - 4/p}$. Since $|E| \leq n^2$, this
    gives $|S_C| < 2n^{4/p}$, from which
    $|S_C|^2 < 2n^{4/p} \cdot 2|E|/n^{2 - 4/p} = 4|E|/n^{2 - 8/p}$.
    As there are at most $n^{2/p}/\beta$ clusters of size
    $|\vcluster| \geq \beta n^{1 - 2/p}$, we have that the total
    number of edges in $S_C$ over all such $C$ is at most
    $
        (n^{2/p}/\beta) \cdot (4|E|/n^{2 - 8/p})
        = (4/\beta)|E|n^{2 - 10/p}
    $,
    which, for $p \geq 5$, is at most $(4/\beta)|E|$.
\end{proof}

We now focus on listing cliques with an edge in
$E(\vlist, \vlist) \setminus E(S_C, S_C)$ for some $C$ in
$n^{1 - 2/p \rc}$ rounds, using
Lemma~\ref{lem:congest-expander-find-clique}. We show how to deliver
the necessary input edges
$\ep \subseteq E(V \setminus \vlist, V \setminus \vlist)$. Recall that
$K_p$-compatible clusters operate on directed edges; each time a
vertex from outside the cluster sends an edge into $C$, it sends both
directed copies of that edge. This at most doubles the number of edges
received by $C$.

\begin{lemma}%
    \label{lem:send-e-prime}
    For all $v \in S_C^*$ adjacent to a $u \in \vlist \setminus S_C$,
    and for all $v \in (V \setminus \vlist) \setminus S_C^*$, all of
    the edges of $v$ into $\vlist$ can be made known to some
    $u \in \vlist$ so that all $u \in \vlist$ hold
    $O(n^{1 - 2/p} \cdot \deg_C(u))$ such edges. This can be done in
    $O(n^{1 - 2/p})$ rounds and in parallel for all $v$ and $C$.
\end{lemma}
\begin{proof}
    We observe that, for $v \in \vlist \setminus S_C$, we have
    $\deg_{S_C^*}(v) \leq n^{1 - 2/p}$, so that $v$ can learn all
    edges $E(N(v) \cap S_C^*, N(v) \cap S_C^*)$ in $O(n^{1 - 2/p})$
    rounds again similarly to the exhaustive-search approach of
    Lemma~\ref{thm:exhaustive-search}.

    Also, for each $u \in (V \setminus \vlist) \setminus S_C^*$, we
    have $\deg_C(u) \geq \deg_{V \setminus \vlist}(u)/n^{1 - 2/p}$, so
    $u$ can make each of its edges known to some vertex $v \in \vlist$
    by partitioning its neighborhood into disjoint sets of at most
    $n^{1 - 2/p}$ edges and sending each such set to a different
    $v \in \vlist$. This takes $O(n^{1 - 2/p})$ rounds and can
    be done in parallel for all such $u$ and $C$.

    During this process, each $u \in \vlist$ learns
    $\tilde{O}(n^{1 - 2/p})$ edges in $\ep$ along each of its edges,
    so that the total number of such edges known to any vertex $u$ is
    $
        \tilde{O}(n^{1 - 2/p} \cdot \deg(u))
        = \tilde{O}(n^{1 - 2/p} \cdot \comdeg u)
    $.
\end{proof}

After the process of Lemma~\ref{lem:send-e-prime}, for all
$v \in \vlist \setminus S_C$, each edge between its neighbors is known
to some vertex in $\vlist$. We now wish to use the listing algorithm
of Lemma~\ref{lem:congest-expander-find-clique} to list all instances
of $K_p$ with an edge in $E(\vlist \setminus S_C, \vlist)$. We prove
that most edges are in a cluster that satisfies the conditions of a
$K_p$-compatible cluster, so the remaining clusters can be deferred to
later rounds.

Let $\gamma > 1$ be a constant that we fix later. If
$|E(\vlist, \vcluster)|/|\vlist| \leq |\ep|/(\gamma n)$, then we say
that $C$ is an \emph{overloaded cluster}. Let $E\suhi$ be the set of
all edges in overloaded clusters (i.e., it is the union of all $E_C$
over all overloaded $C$). The following lemma shows that $E\suhi$ can
be deferred to subsequent rounds.

\begin{lemma}%
    \label{lem:few-low-average}
    $|E\suhi| \leq 2|E|/\gamma$
\end{lemma}
\begin{proof}
    Let $\mu'$ denote the average degree of the graph. Since
    $|E\suhi| \leq |E|$, by considering all overloaded clusters, we
    consider vertices with average communication degree inside their
    cluster $|E(\vlist, \vcluster)|/|\vlist|$ bounded above by
    $\mu'/\gamma$, and therefore the number of edges within such
    clusters is at most $\mu' n/\gamma$. Since the total number of
    edges in the graph is $\mu' n/2$, this gives that
    $|E\suhi| \leq 2|E|/\gamma$ as desired.
\end{proof}

The last piece of input needed for the clusters $C$ to be
$K_p$-compatible clusters is that for each $u \in V$, a unique
$v \in \vlist$ holds $\sentdeg u$, the number of edges sent into $C$
by $u$. The general idea is to serially iterate the clusters and have
each $v \in V$ tell each of its neighbors whether it sent the edge
between them to each cluster $C$. As we only consider sufficiently
large clusters, the number of such clusters is small.

\begin{lemma}%
\label{lem:kp-import-sentdeg}
    There exists a deterministic \congest algorithm on $G$ that
    guarantees, in $O(n^{1 - 2/p})$ rounds, for all clusters $C$ with
    $|\vcluster| > \beta n^{1 - 2/p}$, for each $u \in V$ that is a
    tail of at least one $e \in \ebar \cup \ep$ sent into $C$, exactly
    one $v \in \vlist$ holds $\sentdeg u$.
\end{lemma}

\begin{proof}
    We begin by ensuring that each cluster $C$ with
    $|V_C| > \beta n^{1 - 2/p}$ has a unique identifier that can be
    encoded in $O(1)$ messages and is known to all vertices in that
    cluster. 
    Due to Theorem~\ref{thm:diameter-conductance}, we can
    construct a spanning tree of $C$ in $\tilde{O}(1)$ rounds and use
    it to compute and distribute the minimum id of a vertex in
    $\vlist$; call this id the \emph{representative} $r_C$ of $C$.

    Now, each $v \in V$ sends over each edge $e$ incident to $v$,
    $e = (v, u)$, the representative $r_C$ of each cluster $C$ to
    which it sent that edge. If $v \in \vlist$ for some $C$ and knows
    an edge $(v, u) \in \ebar$, it also sends $r_C$ to $u$.
    Since we only consider $C$ with $|\vcluster| > \beta n^{1 - 2/p}$,
    there are fewer than $n^{2/p} \le n^{1/2}$ such clusters, so that
    all such representatives can be sent in $O(n^{1/2})$ rounds. Now
    each $v \in V$ knows for a given $r_C$ and a given $e = (v, u)$
    whether $e$ was sent to $C$ as part of $\ebar \cup \ep$: in
    particular, $e$ was sent to $C$ iff $v$ itself sent $e$ to $C$, or
    if $u$ sent $e$ to $C$ and then transmitted $r_C$ to $v$
    over $e$. Therefore, each $v \in \vlist$ can compute $\sentdeg v$.

    To guarantee that each $C$ receives only one copy of $\sentdeg v$,
    $v$ decides which vertex $u \in V$ will send $\sentdeg v$ to $v$.
    Let $V_{send}$ be the set of vertices $u \in N(v) \cup \{v\}$ that
    sent an edge $e$ incident to $v$ to $C$, and let $w$ be the vertex
    in $V_{send}$ with minimum id. Then if $w = v$, $v$ has sent one
    of its edges to some $u \in \vlist$ already and can send
    $\sentdeg v$ to that same $u$. Otherwise, $v$ sends $\sentdeg u$
    to $w$ and $w$ sends $\sentdeg v$ to the same $u \in \vlist$ to
    which it sent the edge $(w, v)$.

    Each $v$ sends and receives at most two messages for each
    edge it sent to some $C$, thus this takes
    $O(n^{1 - 2/p})$ rounds by Lemma~\ref{lem:send-e-prime}.
\end{proof}

We now combine all the above, as well as choose values for $\epsilon$,
$\beta$, and $\gamma$, to prove Lemma~\ref{lem:kp-wrapper}.

\begin{proof}[Proof of Lemma~\ref{lem:kp-wrapper}]
    After running the expander decomposition of
    Theorem~\ref{thm:cs-decomposition}, invoke
    \Cref{lem:remaining-edges-small,lem:list-low-degree,%
    lem:sc2-is-small,lem:send-e-prime,lem:few-low-average}, which
    guarantee that, if we list all instances of $K_p$ with an edge in
    $E(\vlist, \vlist \setminus S_C)$ for all not overloaded $C$
    satisfying $|\vcluster| > \beta n^{1 - 2/p}$ in $n^{1 - 2/p \rc}$
    rounds, we arrive at a recursive algorithm that lists all $K_p$,
    in $n^{1 - 2/p \rc}$ rounds per invocation, and removes at least
    $(1 - 3\epsilon - 4/\beta - 2/\gamma)|E|$ edges from $G$ per each
    invocation.

    We show that clusters $C$ of size
    $|\vcluster| > \beta n^{1 - 2/p}$ that are not overloaded satisfy
    the conditions of a $K_p$-compatible cluster, so we can list all
    all $K_p$ cliques with an edge in
    $E(\vlist, \vlist \setminus S_C)$ for such $C$ in
    $n^{1 - 2/p \rc}$ rounds. Since the clusters $C$ are
    edge-disjoint and the sparsity-aware listing algorithm only uses
    the edges $E_C$ for communication, we can do so in parallel for
    all $C$. We let $\ebar = E(\vlist, V \setminus \vlist)$ and
    $\ep$ be the set of edges sent as input in
    Lemma~\ref{lem:send-e-prime}. By the definition of $\vlist$,
    all $v \in \vlist$ satisfy $\deg_C(v) \geq \frac12\deg(v)$ and
    $\deg_C(v) > \beta n^{1 - 2/p}$, thus we have that $\ebar$
    satisfies $|E(\vlist, V) \cap \ebar| \in O(\comdeg v)$ and
    the minimum degree is at least
    $\beta n^{1 - 2/p} \ge K^{1 - 2/p}$. By
    Lemma~\ref{lem:send-e-prime}, each $v \in \vlist$ learns
    $O(n^{1 - 2/p})$ edges in $\ep$. Also, exactly one $v \in \vlist$
    learns $\sentdeg u$ for each required $u \in V$ by
    Lemma~\ref{lem:kp-import-sentdeg}. Finally, since $C$ is not
    overloaded, the condition $\mu \in \Omega(|\ep|/n)$ holds.

    The depth of the recursion is logarithmic in $n$, as letting
    $\epsilon = 1/18$, $\beta = 24$, and $\gamma = 12$ gives that each
    iteration of the algorithm eliminates at least half of the edges
    from $G$.
\end{proof}

\subsection{\texorpdfstring{$
    \boldsymbol{K_p}
$}{K\_p} Listing for \texorpdfstring{$
    \boldsymbol{p = 4}
$}{p \textgreater\ 4}}%
\label{sec:wrapper-clique-4}

The listing algorithm for $p > 4$ in
Section~\ref{sec:wrapper-clique-5} cannot be applied for $p = 4$ as
the bound in Lemma~\ref{lem:sc2-is-small} requires $p \geq 5$. As a
result, we can no longer defer all edges between bad vertices to
future rounds and must instead list cliques containing these edges.

We begin as in Section~\ref{sec:wrapper-clique-5} by fixing constants
$\beta > 1$ and $\epsilon$ and producing communication clusters
from an expander decomposition of
$G$. We define $\vinnerof{C_i}$ and $E_i^-$ as in
Section~\ref{sec:preliminaries} and let
$E_i^+ = E_i \cup E(\vinnerof{C_i}, \vinnerof{C_i})$. We again let
$C_i = G[E_i^+]$, so that there is some set of clusters
$\mathbb{C}$ such that $C \in \mathbb{C}$ is a cluster with vertices
$\vcluster$, edges $E_C$, and subsets of vertices $\vinner$
and $\vlist$ satisfying $\deg_C(v) \geq \deg_G(v)$ for all
$v \in \vinner$ and $\deg_C(v) \geq \beta n^{1/2}$ for all
$v \in \vlist$. We later fix $\epsilon$ and $\beta$ so that we
can focus on listing only those cliques with an edge in
$E(\vinner, \vinner)$ for some $C$.

We now cover the entire graph $G$ with expanders in the following
manner. Let
$E_{rem} = E \setminus \bigcup_{C \in \mathbb{C}} E(\vinner, \vinner)$
be the set of edges in $G$ that are not within the sets of edges whose
cliques we seek to list. Recursively perform the expander
decomposition on the graph induced by $E_{rem}$, and let the set of
all clusters in all expander decompositions created in this way be
$\mathbb{C}^*$. We now show that we can use the expander rounding
algorithm of Theorem~\ref{thm:cs-routing} in parallel over all
clusters $C \in \mathbb{C}^*$ with only a logarithmic factor in
overhead.

\begin{lemma}%
\label{lem:log-recursive-overhead}
    If $\epsilon \leq 1/4$, then each edge $e \in E$ is in $E_C$ for
    at most $O(\log n)$ clusters $C \in \mathbb{C}^*$, and each vertex
    $v \in V$ is in $\vlist$ for at most $O(\log n)$ clusters
    $C \in \mathbb{C}^*$.
\end{lemma}

\begin{proof}
    By
    Lemma~\ref{lem:remaining-edges-small}, if $\epsilon \le 1/4$,
    $|E_{rem}| < \frac12|E|$. Thus, the depth of the recursive
    application of expander decomposition to $G$ is logarithmic, and
    there are $O(\log n)$ distinct decompositions. In each
    decomposition, each edge is in at most one set $E_C$ and each
    vertex is in at most one set $\vlist$, completing the proof.
\end{proof}

As in the $p > 4$ case, we use exhaustive search via
Lemma~\ref{lem:list-low-degree} to list all cliques that contain a
vertex with degree at most $2\beta n^{1/2}$. In every
$C \in \mathbb{C}^*$, if $|\vcluster| < \beta n^{1/2}$, then the
exhaustive search lists all cliques containing some $v \in \vinner$,
therefore, from now on we only consider $C$ satisfying
$|\vcluster| > \beta n^{1/2}$. Since by
Lemma~\ref{lem:log-recursive-overhead} the depth of the recursive
decomposition is logarithmic, and each iteration of the decomposition
can produce at most $n^{1/2}/\beta$ clusters of size
$|\vlist| > \beta n^{1/2}$, the total number of such clusters is
$\tilde{O}(n^{1/2})$.

Then, note that all $K_4$ that lie entirely in some cluster
$C \in \mathbb{C}$ can be listed in $n^{1/2 \rc}$ rounds by
applying the assumed listing algorithm for $K_4$-compatible clusters
to $C$ with $\ep = \ebar = \emptyset$.

We now seek to list remaining $K_4$ with vertices
$\{v_1, v_2, v_3, v_4\}$ where
$\{v_1, v_2\} \in E(\vlist, \vlist)$ for some $C \in \mathbb{C}$,
and
$\{v_3, v_4\} \in E(\vlistof{C^*}, \vlistof{C^*})$ for some
$C^* \in \mathbb{C}^* \setminus \{C\}$. For each $K_4$ across some
pair $(C, C^*) \in \mathbb{C} \times \mathbb{C}^*$, we ensure
either $C$ or $C^*$ learns all of the edges in that $K_4$ so
that it can be listed using the assumed listing algorithm for
$K_4$-compatible clusters. For two distinct clusters $C$ and $C^*$, we
define 
\[
    S_{C^* \to C}^* = \{
        u \in \vlistof{C^*} :
        1 \leq \deg_C(u) < \deg_{\vlistof{C^*}}(u) / \sqrt{n}
    \}
\]
\[
    S_{C \to C^*} = \{
        u \in \vlist :
        \deg_{S_{C^* \to C}^*}(u) > \sqrt{n}
    \}
\]
which are analogs of $S_C^*$ and $S_C$ from
Section~\ref{sec:wrapper-clique-5}. The algorithm runs in three
parts, listing all $K_4$ such that (1): $v_3$ or $v_4$ is in
$\vlistof{C^*} \setminus S_{C^* \to C}^*$, (2): $v_1$ or $v_2$ is in
$\vlist \setminus S_{C \to C^*}$ and $v_3, v_4 \in S_{C^* \to C}^*$,
or (3) $v_1, v_2 \in S_{C \to C^*}$ and
$v_3, v_4 \in S_{C^* \to C}^*$.

\paragraph{\textbf{Part 1.}} For each $C^* \in \mathbb{C}^*$ and
$C \in \mathbb{C} \setminus C^*$, each vertex
$u \in \vlistof{C^*} \setminus S_{C^* \to C}^*$ that has
$\deg_{\vlist}(u) > 0$ sends its neighborhood $N_{\vlistof{C^*}}(u)$
to $C$. This takes $n^{1/2 \rc}$ rounds by the definition of
$S_{C^* \to C}^*$ and Lemma~\ref{lem:log-recursive-overhead}. Each
$C \in \mathbb{C}$ then proceeds similarly to the case of $p > 4$.
Denote by $\ep \subseteq E(V \setminus \vlist, V \setminus \vlist)$
the edges not incident to $\vlist$ that are distributed among $\vlist$
after the step above. $C$ computes $|\ep|$ and
$|E(\vlist, \vcluster)|/|\vlist|$ in $\operatorname{polylog}(n)$
rounds. We use the definition that $C$ is an \emph{overloaded cluster}
if $|E(\vlist, \vcluster)|/|\vlist| \leq |\ep|/(\gamma n)$ and defer
all such $C$ to future iterations. By Lemma~\ref{lem:few-low-average},
the number of edges in all such clusters is at most $2|E|/\gamma$. For
$C$ that are not overloaded, 
$\ebar = E(\vlist, V \setminus \vlist)$. Below, in
Lemma~\ref{lem:k4-import-sentdeg}, we show it is possible for each
$C \in \mathbb{C}^*$ to learn $\sentdeg u$ for all relevant $u$, at
which point $C$ satisfies the definition of a $K_4$-compatible cluster
and can list all $K_4$ in $G[E(\vlist, \vlist) \cup \ebar \cup \ep]$
with an edge in $E(\vlist, \vlist)$ in $n^{1/2 \rc}$ rounds.

\paragraph{\textbf{Part 2.}} For each $C \in \mathbb{C}$ and
$C^* \in \mathbb{C}^* \setminus \{C\}$, each vertex
$u \in \vlist \setminus S_{C \to C^*}$ learns all edges in
${N_{S^*_{C^* \to C}}(u)}^2$ from $C^*$. This takes $n^{1/2 \rc}$
rounds by the definition of $S_{C \to C^*}$ and
Lemma~\ref{lem:log-recursive-overhead}. $C$ then takes
$\ebar = E(\vlist, V \setminus \vlist)$ and
$\ep \subseteq E(V \setminus \vlist, V \setminus \vlist)$ the edges
not incident to $\vlist$ now distributed among $\vlist$.
Lemma~\ref{lem:k4-import-sentdeg} shows that it is possible for each
$C \in \mathbb{C}^*$ to learn $\sentdeg u$ for all relevant $u$, at
which point $C$ satisfies the definition of a $K_4$-compatible cluster
and can list all $K_4$ in $G[E(\vlist, \vlist) \cup \ebar \cup \ep]$
with an edge in $E(\vlist, \vlist)$ in $n^{1/2 \rc}$ rounds.

\paragraph{\textbf{Part 3.}} We now wish to list cliques
$\{v_1, v_2, v_3, v_4\}$ for which $v_1, v_2 \in S_{C \to C^*}$ and
$v_3, v_4 \in S_{C^* \to C}^*$ for some $C \in \mathbb{C}$ and
$C^* \in \mathbb{C}^* \setminus \{C\}$ as follows. Each vertex
$u \in S_{C \to C^*}$ sends the set of edges
$\{\{u, v\} : v \in N_{S_{C \to C^*}}(u)\}$ to $C^*$ by dividing it
into $\deg_{S_{C^* \to C}^*}(u)$ subsets of approximately equal size
and sending each subset to one of its neighbors in $S_{C^* \to C}^*$.
By the definition of $S_{C \to C^*}$ and
Lemma~\ref{lem:log-recursive-overhead}, this takes $n^{1/2 \rc}$
rounds. Each cluster $C^* \in \mathbb{C}$ then takes
$
    \ebar = \bigcup_{C \in \mathbb{C} \setminus \{C^*\}}
        E(S_{C \to C^*}, S_{C^* \to C}^*)
$
and
$
    \ep \subseteq E(
        V \setminus \vlistof{C^*},
        V \setminus \vlistof{C^*}
    )
$
the edges not incident to vertices in $\vlistof{C^*}$ now distributed
amongst the vertices in $C^*$. Lemma~\ref{lem:k4-import-sentdeg} shows
that it is possible for each $C^* \in \mathbb{C}^*$ to learn
$\sentdeg u$ for all relevant $u$, at which point $C^*$ satisfies the
definition of a $K_4$-compatible cluster and can list all $K_4$ in
$G[E(\vlistof{C^*}, \vlistof{C^*}) \cup \ebar \cup \ep]$ with an edge
in $E(\vlistof{C^*}, \vlistof{C^*})$ in $n^{1/2 \rc}$ rounds.

We now prove that the clusters tasked with listing in Parts~1--3 above
are in fact $K_4$-compatible clusters. The first lemma shows that each
$C \in \mathbb{C}^*$ can learn the necessary values of $\sentdeg u$
for the edge sets it receives. This is the analogue of
Lemma~\ref{lem:kp-import-sentdeg} from the algorithm for $p > 4$.

\begin{lemma}%
\label{lem:k4-import-sentdeg}
    There exists a deterministic \congest algorithm on $G$ that, in
    $n^{1/2 \rc}$ rounds, guarantees for all clusters
    $C \in \mathbb{C^*}$ with $|\vcluster| > \beta n^{1/2}$, and each
    $u \in V$ that is a tail of at least one $e \in \ebar \cup \ep$
    sent into $C$ in any of Parts~1,~2, or~3, that exactly one
    $v \in \vlist$ holds $\sentdeg u$.
\end{lemma}

\begin{proof}
    We roughly follow Lemma~\ref{lem:kp-import-sentdeg}, with the new
    complication being that clusters $\mathbb{C}^*$ are no longer
    vertex-disjoint, thus the minimum id of a vertex in $\vlist$ is no
    longer a unique identifier of $C$. Instead, we let the
    representative $r_C$ of $C$ be the pair $(v, i)$, where $v$ is the
    lowest-numbered vertex in $\vlist$ and $i$ is the iteration number
    of the recursive partitioning process at which $C$ was created. As
    within each invocation of the expander decomposition the clusters
    produced are vertex-disjoint, $r_C$ now uniquely identifies $C$.
    The number of clusters with $\vlist \neq \emptyset$ is still
    $\tilde{O}(n^{1/2})$, thus we can use the algorithm of
    Lemma~\ref{lem:kp-import-sentdeg}.
\end{proof}

The next lemmas show $\ebar$ and $\ep$
are sent to $C$ while satisfying the
$K_4$-compatible cluster conditions.

\begin{lemma}
    \label{lem:low-c*-neighborhood}
    For each $C^* \in \mathbb{C}^*$, each $v \in \vlistof{C^*}$ is
    incident to $O(\deg_{C^*}(v))$ edges from $\ebar$.
\end{lemma}

\begin{proof}
    For each $v \in \vlistof{C^*}$ and a cluster
    $C \in \mathbb{C} \setminus \{C^*\}$ such that
    $v \in S_{C^* \to C}^*$, we have that
    $\deg_C(v) < \deg_{\vlistof{C^*}}(v)/\sqrt{n}$. The total number
    of edges $\ebar$ incident to $v \in \vlistof{C^*}$ is at
    most
    \begin{align*}
        \sum_{
            C \in \mathbb{C} \setminus \{C^*\},
            v \in S^*_{C^* \to C}
        } \deg_{S_{C \to C^*}}(v)
        &\leq O(\sqrt{n}) \cdot 
            \max\{
                \deg_{S_{C \to C^*}}(v) :
                C \in \mathbb{C} \setminus \{C^*\},
                v \in S^*_{C^* \to C}
            \} \\
        &\leq O(\deg_{C^*}(v))
    \end{align*}
    For the first inequality, we observe that
    $|\mathbb{C}| \in O(\sqrt{n})$, since we have already eliminated
    by brute force all of the clusters $C \in \mathbb{C}$ with size
    $|\vcluster| < 2\beta n^{1/2}$. The second follows from the
    definition of $S_{C^* \to C}^*$.
\end{proof}

\begin{lemma}
    \label{lem:c*-edges-received}
    For each $C^* \in \mathbb{C}^*$, each $v \in \vlistof{C^*}$
    receives $O(\sqrt{n} \cdot \deg_{C^*}(v))$ edges at the first step
    of part 3.
\end{lemma}

\begin{proof}
    $v \in \vlistof{C^*}$ receives messages from
    $u \in S_{C \to C^*}$ only if $v \in S_{C^* \to C}^*$. Thus,
    $v$ receives $|N_{S_{C \to C^*}}(u)|/\sqrt{n} \leq \sqrt{n}$ edges
    in $E_C$ from $u \in \vlist$, so the total number of
    edges received is at most
    \begin{align*}
        \sqrt{n} \cdot \sum_{
            C \in \mathbb{C} \setminus \{C^*\},
            v \in S^*_{C^* \to C}
        } \deg_{S_{C \to C^*}}(v)
        &\leq \sqrt{n} \cdot \left(
            O(\sqrt{n}) \cdot
            \max\{
                \deg_{S_{C \to C^*}}(v) :
                C \in \mathbb{C} \setminus \{C^*\},
                v \in S^*_{C^* \to C}
            \}
        \right) \\
        &\leq O(\sqrt{n} \cdot \deg_{C^*}(v))
    \end{align*}
    again using the property $|\mathbb{C}| \in \tilde{O}(n^{1/2})$ and
    the definition of $S_{C^* \to C}^*$.
\end{proof}

Finally, we show that the total size of $|\ep|$ is not too large.

\begin{lemma}
    \label{lem:c*-not-overloaded}
    For each $C^* \in \mathbb{C}^*$, the average degree of $C^*$ is at
    least $\max_C\{|S_{C \to C^*}|\}$, implying
    $
        |E(\vlistof{C^*}, \vclusterof{C^*})|/|\vlistof{C^*}|
        = \Omega(|\ep|/n)
    $.
\end{lemma}

\begin{proof}
    From the definition of $S_{C \to C^*}$, each vertex
    $v \in S_{C \to C^*}$ is incident to at most $\sqrt{n}$ edges from
    $E(\vlist, S_{C^* \to C}^*)$. By the definition of
    $S_{C^* \to C}^*$, we can associate $\sqrt{n}$ edges in
    $E(S_{C^* \to C}^*, \vclusterof{C^*})$ for each edge in
    $v \in S_{C \to C^*}$. Therefore, the number of edges in
    $E(S_{C^* \to C}^*, V)$ is at least
    \[
        \sqrt{n} \cdot \sqrt{n} \cdot |S_{C \to C^*}|
        = n \cdot |S_{C \to C^*}|
    \]
    As this is true for all
    $C \in \mathbb{C} \setminus \{C^*\}$, the average
    degree in $C^*$ is at least $\max_C\{|S_{C \to C^*}|\}$.
\end{proof}

We can now complete the proof of optimal clique listing.

\begin{proof}[Proof of Lemma~\ref{lem:k4-wrapper}]
    By Lemmas~\ref{lem:k4-import-sentdeg},%
~\ref{lem:low-c*-neighborhood},%
~\ref{lem:c*-edges-received}, and~\ref{lem:c*-not-overloaded}, the
    clusters tasked with listing instances of $K_4$ in Parts~1--3 of
    the algorithm above satisfy the conditions of $K_4$-compatible
    clusters so that all $K_4$ containing an edge in
    $E(\vinner, \vinner)$ for a non-overloaded cluster
    $C \in \mathbb{C}$ can be listed in $n^{1/2 \rc}$
    rounds. Letting $\epsilon = 1/12$ and $\gamma = 4$, we have by
    Lemmas~\ref{lem:remaining-edges-small}
    and~\ref{lem:few-low-average} that at most $|E|/2$ edges are
    deferred to later rounds, so that the depth of recursion is
    logarithmic. Therefore, the overall runtime is
    $n^{1/2 \rc}$.
\end{proof}




\paragraph{Acknowledgements: This project was partially supported by
the European Union's Horizon 2020 Research and Innovation Programme
under grant agreement no. 755839.}

\bibliography{partition-tree-notes}

\end{document}